\documentclass[11pt,journal,onecolumn]{IEEEtran}

\usepackage{cite}

\usepackage[]{graphicx}
\graphicspath{{figures/}}
\DeclareGraphicsExtensions{.eps,.pdf}
\usepackage{caption}
\usepackage{subcaption}

\usepackage{amsmath,amssymb,amsthm,dsfont}
\usepackage{enumerate}
\usepackage[mathscr]{euscript}
\usepackage{stmaryrd}

\usepackage{mathtools}
\usepackage{verbatim}
\DeclarePairedDelimiter \norm{\lVert}{\rVert}%
\newtheorem{theorem}{Theorem}

\newtheorem{lemma}{Lemma}
\newtheorem{proposition}{Proposition}

\theoremstyle{remark}
\newtheorem*{remark}{Remark}
\IEEEoverridecommandlockouts

\makeatletter
\def\@xfootnote[#1]{%
	\protected@xdef\@thefnmark{#1}%
	\@footnotemark\@footnotetext}
\makeatother

\allowdisplaybreaks[4] 
\usepackage{multirow}
\usepackage{ctable}

\usepackage{placeins} 

\usepackage{cite}
\usepackage[
colorlinks=true,
citecolor=black,
linkcolor=black,
urlcolor=black,
]{hyperref}

\usepackage[]{color}

\begin{document}

\title{Strong Coordination over Noisy Channels}

\author{ \thanks{This work was supported by NSF grants CCF-1440014, CCF-1439465.}
	\thanks{This work was presented in part at the IEEE International Symposium on Information Theory (ISIT 2017), Aachen, Germany \cite{OVK17} and at the 55th Annual Allerton Conference, Monticello, IL, USA \cite{OKV17}.}
\IEEEauthorblockN{Sarah A. Obead, J\"{o}rg Kliewer\\}
\IEEEauthorblockA{Department of Electrical and Computer Engineering\\
New Jersey Institute of Technology\\
Newark, New Jersey 07102\\
Email:sao23@njit.edu, jkliewer@njit.edu\\}
\and \vspace{3 ex}
\IEEEauthorblockN{Badri N. Vellambi\\}
\IEEEauthorblockA{Research School of Computer Science\\
	Australian National University\\
	Acton, Australia 2601\\
	Email: badri.n.vellambi@ieee.org}
}

\maketitle

\begin{abstract}
	
We study the problem of strong coordination of the actions of two nodes $\mathsf X$ and $\mathsf Y$ that communicate over a discrete memoryless channel (DMC) such that the actions follow a prescribed joint probability distribution. We propose two novel random coding schemes and a polar coding scheme for this noisy strong coordination problem, and derive inner bounds for the respective strong coordination capacity region. The first scheme is a joint coordination-channel coding scheme that utilizes the randomness provided by the communication channel to reduce the amount of local randomness required to generate the sequence of actions at Node $\mathsf Y$. Based on this random coding scheme, we provide a characterization of the capacity region for two special cases of the noisy strong coordination setup, namely, when the actions at Node $\mathsf Y$ are determined by Node $\mathsf X$ and when the DMC is a deterministic channel.
The second scheme exploits separate coordination and channel coding where local randomness is extracted from the channel after decoding. The third scheme is a joint coordination-channel polar coding scheme for strong coordination. We show that polar codes are able to achieve the established inner bound to the strong noisy coordination capacity region and thus provide a constructive alternative to a random coding proof. Our polar coding scheme also offers a constructive solution to a channel simulation problem where a DMC and shared randomness are employed together to simulate another DMC. Finally, by leveraging the random coding results for this problem, we present an example in which the proposed joint scheme is able to strictly outperform the separate scheme in terms of achievable communication rate for the same amount of injected randomness into both systems. Thus, we establish the sub-optimality of the separation of strong coordination and channel coding with respect to the communication rate over the DMC in this problem.

\end{abstract}


\section{Introduction}

A fundamental problem in decentralized networks is to coordinate
activities of different nodes with the goal of reaching a state of
agreement. The problem of communication-based coordination of multi-node systems
arises in numerous applications including autonomous robots, smart
traffic control, and distributed computing such as distributed games and
grid computing \cite{cuff2010coordination}. Coordination is understood to be the ability to arrive at a prescribed joint distribution of actions at all nodes in the network.
Several theoretical and applied studies on
multi-node coordination have targeted  questions on  how nodes
exchange information and how their actions can be correlated to achieve a
desired overall behavior. Two types of coordination have been addressed in
the literature -- \emph{empirical} coordination where the normalized histogram of induced
 joint actions is required to be close to a prescribed target distribution, and
\emph{strong} coordination, where the induced sequence of joint actions of all the
nodes is required to be statistically close (i.e., nearly indistinguishable) from a chosen target
probability mass function (pmf).

Recently, a significant amount of work has been devoted to finding the
capacity regions of various coordination network problems based on both
empirical and strong coordination \cite{soljanin2002compressing,cuff2010coordination,cuff2013distributed, BK14, VKB16, bereyhi2013empirical}. Bounds on the capacity region for the point-to-point case  were obtained in \cite{gohari2011generating} under the assumption that the nodes communicate in a bidirectional fashion in order to achieve coordination. A similar framework was adopted and improved in \cite{yassaee2015channel}. In
\cite{BK14,haddadpour2012coordination,bereyhi2013empirical}, the authors addressed inner and outer bounds for the capacity region of a three-terminal network in the presence of a relay. The work of \cite{BK14} was later extended in \cite{bloch2013strong,VKB16} to derive a  precise characterization of the strong coordination region for multi-hop networks. 

While the majority of recent works on coordination have considered noise-free communication channels, coordination over noisy channels has received only little attention in 
the literature so far. However, notable exceptions are \cite{CS11,HYBGA17,CLTB18}. In \cite{CS11}, joint empirical coordination of the channel inputs/outputs of a noisy communication channel with source and reproduction sequences is considered, and in \cite{HYBGA17}, the notion of strong coordination is used to simulate a discrete memoryless channel (DMC) via another channel. Recently, the authors of \cite{CLTB18}
explored the strong coordination variant of the problem investigated in \cite{CS11} when two-sided channel state information is present and side information is available at the decoder.

As an alternative to the impracticalities of random coding, solutions for empirical and strong coordination problems have been proposed based on low-complexity polar-codes introduced by Arikan \cite{ChPolarztion2009Arikan,SPolarztion2010Arikan}. For example, polar coding for
	strong point-to-point coordination is addressed in
	\cite{bloch2012strong,chou2016empirical}, and for empirical coordination
	in cascade networks in \cite{blasco2012polar}, respectively. The
	only existing design of polar codes for the
	noisy empirical coordination case 
	\cite{PCempirical2016} is based on the joint source-channel
	coordination approach in \cite{CS11}. A construction based on polar
        codes for the noisy \emph{strong} coordination problem has been
        first presented in our previous conference work \cite{OKV17}, which is
        part of this  paper.

In this work, we consider the point-to-point coordination setup illustrated in Fig.~\ref{fig:P2PCoordination}, where in contrast to \cite{CS11} and \cite{CLTB18} only source and reproduction sequences at two different nodes ($\mathsf X$ and $\mathsf Y$) are coordinated by means of a suitable communication scheme over a DMC. Specifically, we propose two novel achievable coding schemes for
this noisy coordination scenario, derive inner bounds to the underlying strong coordination
capacity region, and provide the capacity region for
  two special cases of the noisy strong coordination setup. In particular, we characterize the capacity region for the cases when the actions at Node $\mathsf Y$ are determined by Node $\mathsf X$ and when the DMC is deterministic. Finally, we design an explicit low-complexity nested polar coding scheme that achieves the inner bound of the point-to-point noisy coordination capacity region.

 The first scheme is a joint coordination channel coding
scheme that utilizes randomness provided by the DMC to reduce the local
randomness required in generating the action sequence at Node $\mathsf Y$ (see
Fig.~\ref{fig:P2PCoordination}). Even though the proposed joint scheme is related to the
scheme in \cite{HYBGA17}, the presented scheme
exhibits a significantly different codebook construction adapted to our coordination
framework. Our scheme requires the quantification of the amount of common
randomness shared by the two nodes as well as the local randomness at each
of the two nodes. To this end, we propose a solution that achieves
  strong coordination over noisy channels via the soft covering principle
  \cite{cuff2013distributed}. Unlike to solutions inspired by
  random-binning techniques, a soft covering based solution is able to quantify the local randomness
  required at the encoder and decoder to generate the correlated action
  sequences. Note that quantifying the amount of local randomness is absent
from the analyses in both \cite{HYBGA17} and \cite{CLTB18}. 
Our second achievable scheme exploits separate coordination and channel coding where local randomness is extracted from the channel after decoding. The third scheme is a joint coordination-channel polar coding scheme that employs nested
codebooks similar to polar codes for the broadcast
channel \cite{PC_BCGoela2015}. We show that our proposed construction provides an equivalent constructive alternative for strong coordination over noisy channels. Here, by equivalent we mean that for every rate point for which one can devise a random joint coordination-channel code, one can also devise a polar coding scheme with significantly lower encoding and decoding complexity. Also, our proposed polar coding scheme employs the soft covering principle \cite{chou2016empirical} and offers a constructive solution to a channel simulation problem, where a DMC is employed to simulate another DMC in the presence of shared randomness~\cite{HYBGA17}.

Lastly, when the noisy channel and the correlation between $X$ to $Y$ are both given by binary symmetric channels (BSCs), we study the effect of the capacity of the noisy channel on
the sum rate of common and local randomness. We conclude this work by showing that a joint coordination-channel coding scheme is able to strictly
outperform a separation-based scheme\footnote{Note that when defining
  separation we also consider the number of channel uses, i.e.,~the
  communication rate, as a quantity of interest besides the communication
  reliability, i.e.,~the probability of decoding error.} in terms of achievable
communication rate if the same amount of randomness is injected into the system in the high-capacity regime for the BSC, i.e.,~$C\rightarrow 1$. This example reveals that separate coordination and channel coding is indeed sub-optimal in the context of strong coordination under the additional constraint of minimizing the communication rate.

The remainder of the paper is organized as follows:
Section~\ref{sec:notation} outlines the notation. The problem of strong
coordination over a noisy communication link is presented in
Section~\ref{sec:problemdef}. We then derive achievability results for the
noisy point-to-point coordination in Section~\ref{sec:JointScheme} for the
joint random-coding scheme and discuss the characterization of the capacity region for two special cases of the noisy strong coordination setup.
Section~\ref{sec:SepCoorRE} presents the separation-based scheme with randomness extraction, and in Section~\ref{sec:PolarCode} we propose a joint
coordination-channel polar code construction and a proof that this
construction achieves the random coding inner bound.
In Section~\ref{sec:example}, we present numerical results for the
  proposed joint and separate coordination and channel coding schemes,
establishing the sub-optimality of the separation-based scheme when the
target joint distribution is described by a doubly binary symmetric source and
the noisy channel by a BSC, respectively.

\section{Notation}\label{sec:notation}
Throughout the paper, we denote a discrete random variable with upper-case
letters (e.g.,~$X$) and its realization with lower case letters
(e.g.,~$x$). The alphabet size of the random variable $X$ is denoted as
$|\mathcal{X}|$. We use {$ \llbracket 1,n \rrbracket$} to denote the set $\{1,\dots,n\}$ for $n\in \mathbb{N}$. Similarly, we use $X_k^{n}$ to denote the finite sequence
$\{X_{k,1},X_{k,2},\ldots,X_{k,n}\}$ and $X_k^{i:j}$ to denote $\{X_{k,i},X_{k,i+1},\ldots,X_{k,j}\}$ such that $1\leq i\leq j \leq n.$ Given ${\cal A}\subset  \llbracket 1,n \rrbracket$, we let $X^n[{\cal A}]$ denote the components $X_i$ such that $i \in {\cal A}.$ 
We use  boldface upper-case letters (e.g.,~${\bf X}$) to denote matrices. We denote the source
polarization transform as  ${\bf G}_n=R{\bf F}^{\otimes n},$ where $R$ is the bit-reversal mapping defined in \cite{ChPolarztion2009Arikan}, 
${\bf F} =\left[\begin{smallmatrix}
1 &  0\\
0 &  0
\end{smallmatrix}\right],$
and ${\bf F}^{\otimes n}$ denotes the $n$-{th} Kronecker product of ${\bf F}.$
The binary entropy function is denoted as $h_2(\cdot)$, and the indicator function by $\mathds{1}(\cdot)$. 
$\mathbb P[A]$ is the
probability that the event $A$ occurs. 
The pmf of the discrete random variable $X$ is denoted as
$P_X(x)$. However, we sometime use the lower case notation (e.g.,~$p_X(x)$)
to distinguish target pmfs or alternative definitions. We let
$\mathbb{D}(P_X(x)||Q_X(x))$ and $\norm{P_X(x)-Q_X(x)}_{{\scriptscriptstyle TV}}$ denote the Kullback-Leibler (KL) divergence and the total variation, respectively, between
two distributions $P_X(x)$ and $Q_X(x)$ defined over an alphabet
$\cal{X}$. Given a pmf $P_X(x)$ we let $\min^{*}_x(P_X) = \min_{x\in {\cal X}} \, \{P_X(x): P_X(x)>0\}$. ${\cal T}_\epsilon^n(P_{X})$ denotes the set of
$\epsilon$-strongly letter-typical sequences of length $n$.
Finally, $P^{n}_{X_1X_2\dots X_k}$ denotes the pmf of $n$ i.i.d.~random variables $X_1,X_2,\dots, X_k$, associated with the pmf $P_{X_1X_2\dots X_k}$.

\section{Problem Definition} \label{sec:problemdef}

\begin{figure}[t!]
	\centering
	{\includegraphics[scale=0.625]{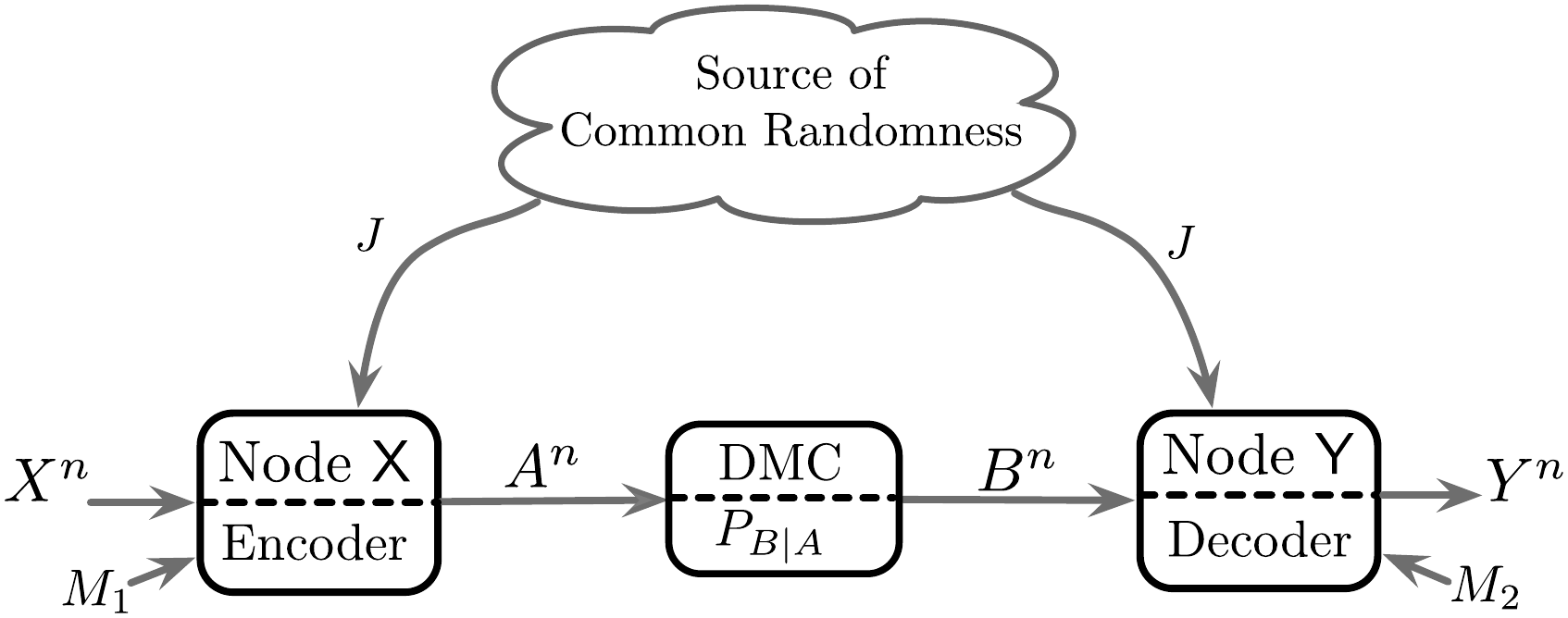}}
	\caption{Point-to-point strong coordination over a DMC.}
	\label{fig:P2PCoordination}
\end{figure}

The point-to-point coordination setup we consider in this work is depicted
in Fig.~\ref{fig:P2PCoordination}. Node $\mathsf X$ receives a sequence of actions
$X^{n}\in \mathcal{X}^{n}$ specified by nature where $X^{n}$ is i.i.d.~according
to a pmf $p_X$. Both nodes have access to shared randomness $J$ at rate
$R_o$ bits/action from a common source, and each node possesses local
randomness $M_\ell$ at rate $\rho_\ell$, $\ell=1,2$. Thus, in designing a \emph{block} scheme to coordinate $n$ actions of the nodes, we assume $J\in  \llbracket 1, 2^{nR_o}\rrbracket $, and $M_\ell\in  \llbracket 1, 2^{n\rho_\ell} \rrbracket $, $k=1,2$, and we wish to communicate a
codeword $A^{n}(I)$ over the DMC $P_{B|A}$ to Node $\mathsf Y$,
where $I$ denotes the (appropriately selected) coordination message. The \emph{codeword} $A^n(I)$ is
constructed based on the input action sequence $X^{n}$, the local randomness
$M_1$ at Node $\mathsf X$, and the common randomness $J$. Node $\mathsf Y$ generates a
sequence of actions $Y^{n}\in \mathcal{Y}^{n}$ based on the received channel output
$B^{n}$, common randomness $J$, and local randomness $M_2$. We assume that the common randomness is independent
of the action specified at Node $\mathsf X$. A tuple $(R_o, \rho_1,\rho_2)$ is deemed \emph{achievable} if for each $\epsilon>0$, there exists $n\in\mathbb{N}$ and a (strong coordination) coding scheme such that the joint
pmf of actions $\hat{P}_{X^{n}Y^{n}}$ induced by this scheme and the $n$-fold product\footnote{This is the joint pmf of $n$ i.i.d. copies of $(X,Y)\sim Q_{XY}$.} of the desired joint pmf  
$Q^{n}_{XY}$ are $\epsilon$\emph{-close} in total variation, i.e.,
\begin{equation}\label{eq:StrngCoorCondtion}
\norm{\hat{P}_{X^nY^n}-Q^{n}_{XY}}_{{\scriptscriptstyle TV}} 
<\epsilon.
\end{equation}

We now present the two achievable coordination schemes.

\section{Joint Coordination Channel Coding}\label{sec:JointScheme}
\subsection{Inner Bound: Achievability}\label{sec:JointSchemeAchievability}
This scheme follows an approach similar to those in \cite{cuff2010coordination,bloch2013strong,BK14,VKB16} where
coordination codes are designed based on allied channel resolvability
problems~\cite{han1993approximation}. The structure of the allied problem pertinent to the coordination problem at hand is given in
Fig.~\ref{fig:StrongCoordinationAllied}. 
The aim of the allied problem is to generate $n$ symbols for two correlated sources
$X^n$ and $Y^n$ whose joint statistics is close to $Q^n_{XY}$ as defined by \eqref{eq:StrngCoorCondtion}. To do so, we employ three independent
and uniformly distributed messages $I$, $K$, and $J$ and two codebooks $\mathscr{A}$
and $\mathscr{C}$ as shown in Fig.~\ref{fig:StrongCoordinationAllied}. To define the two
codebooks, consider auxiliary random variables $A\in\mathcal{A}$ and $C\in\mathcal{C}$ jointly correlated with $(X,Y)$ as  $P_{XYABC}=P_{AC}P_{X|AC}P_{B|A}P_{Y|BC}$ and with marginal distribution $P_{XY} = Q_{XY}.$

\begin{figure}[ht!]
	\centering
	{\includegraphics[scale=0.625]{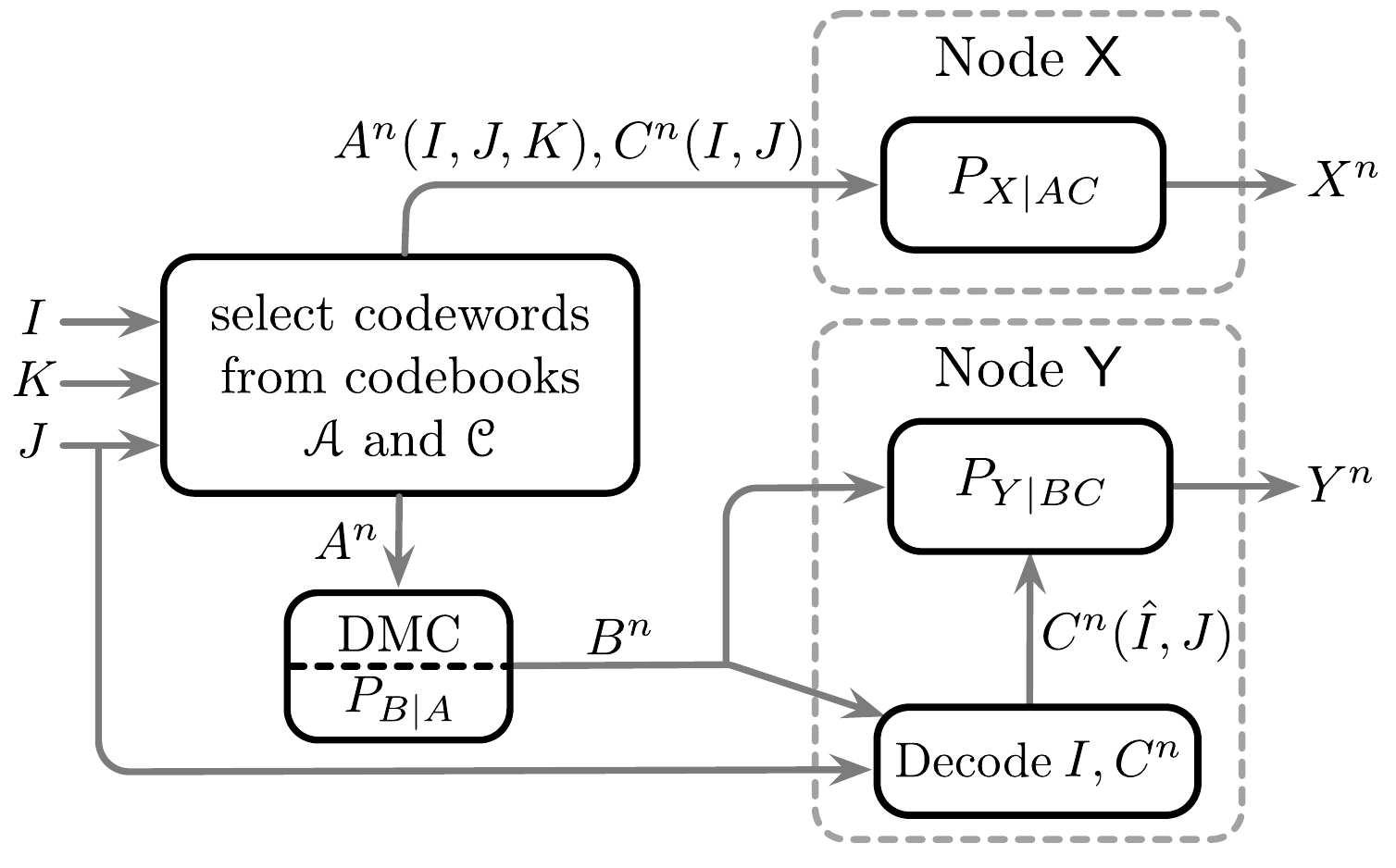}}
	\caption{A joint scheme for the allied problem.} 
	\label{fig:StrongCoordinationAllied}
\end{figure}

From this factorization it can be seen that the scheme consists of two 
\emph{reverse} test channels  $P_{X|AC}$ and $P_{Y|AC}$ used to generate the sources from the
codebooks.
In particular, $P_{Y|AC}=P_{B|A}P_{Y|BC}$, i.e., the randomness of the DMC contributes to the randomized generation of $Y^n$.

Generating $X^n$ and $Y^n$ from $I$, $K$, $J$ represents a complex channel
resolvability problem with the following ingredients:

\begin{enumerate}[i)]
	\item Nested codebooks: Codebook $\mathscr{C}$ of size $2^{n(R_o+R_c)}$ is generated
		i.i.d.~according to pmf $P_{C}$, i.e.,
		$C^n_{ij}\sim  \prod_{l=1}^{n} P_{C}(\cdot)$ for all $(i,j) \in \cal{I}\times \cal{J}$.	Codebook $\mathscr{A}$ is generated by randomly selecting $A^n_{ijk}\sim  \prod_{l=1}^{n}  P_{A|C}(\cdot|{C_{ij}^n})$ for all $(i,j,k) \in \cal{I}\times\cal{J}\times\cal{K}$,  where ${\cal I} \triangleq \llbracket 1,2^{nR_c} \rrbracket$, ${\cal J} \triangleq\llbracket 1,2^{nR_o} \rrbracket$, and ${\cal K} \triangleq \llbracket 1,2^{nR_a} \rrbracket$.
	\item Encoding functions:\\
		$C^n: \llbracket 1, 2^{nR_c} \rrbracket \times \llbracket 1, 2^{nR_o} \rrbracket  \rightarrow \mathcal{C}^n$,\\
		$A^n: \llbracket 1, 2^{nR_c} \rrbracket \times \llbracket 1, 2^{nR_o} \rrbracket \times \llbracket 1, 2^{nR_a} \rrbracket  \rightarrow \mathcal{A}^n$.
	\item Indices: $I,J,K$ are independent and uniformly distributed over ${\cal I}$, ${\cal J}$, and ${\cal K}$ respectively. 
		These indices select the pair of codewords $C^n_{IJ}$ and $A^n_{IJK}$ from codebooks $\mathscr{C}$ and $\mathscr{A}$.
	\item The selected codewords $C^n_{IJ}$ and $A^n_{IJK}$ are then
		passed through DMC $ P_{X|AC}$ at Node
		$\mathsf X$, while at Node $\mathsf Y$, codeword $A^n_{IJK}$ is sent
		through DMC $P_{B|A}$ whose output $B^n$ is used to decode codeword $C^n_{\hat{I}J}$ and both are then passed
		through DMC $P_{Y|BC}$ to obtain $Y^n$. 
\end{enumerate} 

Since the codewords are randomly chosen, the induced joint~pmf of the generated actions and codeword indices in the allied problem is itself a random variable and depends on the random codebook. Given a realization of the codebooks 
\begin{align}
\mathsf C \triangleq (\mathscr{A},\mathscr{C})=\left\{ a_{ijk}^n, c_{ij}^n: \substack{i\in\llbracket 1, 2^{nR_c} \rrbracket\\j\in\llbracket 1, 2^{nR_o} \rrbracket\\ k\in\llbracket 1, 2^{nR_a} \rrbracket}\right\}, \label{eqn-codedefn}
\end{align}
 the code-induced joint~pmf of the actions and codeword indices in the allied problem is given by

\begin{equation}\label{eqn-codeinducedAct}
\mathring{P}_{X^nY^nIJK}(x^n,y^n,i,j,k) \triangleq \frac{P^n_{X|AC}(x^n|a_{ijk}^nc^n_{ij})}{2^{n(R_c+R_o+R_a)}} \Big(\sum_{b^n,\hat i} P^n_{B|A}(b^n|a_{ijk}^n) \mathsf P_{\hat I | B^nJ}(\hat i|b^n,j) P^n_{Y|BC}(y^n|b^nc^n_{\hat ij}) \Big),
\end{equation}	

\noindent where $\mathsf P_{\hat I|B^nJ}$ denotes the operation of decoding the index $I$ using the common randomness and the channel output at Node $\mathsf Y$. Note that the indices for the $C$-codeword that generate $X$ and $Y$ sequences  in \eqref{eqn-codeinducedAct} can be different since the decoding of the index $I$ at Node $\mathsf Y$ may fail. We are done if we accomplish the following tasks: (1) identify conditions on $R_o, R_c,R_a$ under which the code-induced pmf $\mathring P_{X^nY^n}$ is \emph{close} to the design pmf $Q^n_{XY}$ with respect to total variation; and (2) devise a 
strong coordination scheme by inverting the operation at  Node $\mathsf X$ in Fig.~\ref{fig:StrongCoordinationAllied}. This will be done in following sections by subdividing the analysis of the allied problem.\\

\subsubsection{Resolvability constraints} \label{sec:ChResCons}
Assuming that the decoding of $I$ and the codeword $C^n_{IJ}$ occurs perfectly at Node $\mathsf Y$, we see that the code-induced joint pmf induced by the allied scheme for the realization of the codebook $\mathsf C$ in \eqref{eqn-codedefn} is

\begin{equation}\label{eqn-codeinduced}
\check{P}_{X^nY^nIJK}(x^n,y^n,i,j,k)=\frac{P^n_{X|AC}(x^n|a_{ijk}^nc^n_{ij})}{2^{n(R_c+R_o+R_a)}}
\Big( \sum_{b^n} P^n_{B|A}(b^n|a_{ijk}^n)  P^n_{Y|BC}(y^n|b^nc^n_{ij}) \Big).
\end{equation}

The following result quantifies when the above induced distribution is close to the $n$-fold product of the design pmf $Q_{XY}$.  

\begin{lemma}[Resolvability constraints]
	\label{Lma:Reslv}
	The total variation between the code-induced pmf $\check{P}_{X^nY^n}$ in \eqref{eqn-codeinduced} and the desired pmf $Q^n_{XY}$ asymptotically vanishes, i.e., 
	${\mathbb E_\mathsf{C}}\big[ \norm*{\check{P}_{X^nY^n}-Q^n_{XY}}_{{\scriptscriptstyle TV}}\big] \rightarrow 0$ as $n\rightarrow \infty$, if 
	\begin{align}
	R_a+R_o+R_c& > I(XY;AC)\label{eqn-resolve1}\\
	 R_o+R_c &> I(XY;C).\label{eqn-resolve2}
	\end{align}
	 Note that in the above, we let $\mathbb E_\mathsf{C}$ to denote the expectation over the random realization of the codebooks.
\end{lemma}
\begin{proof}
In the following, we drop the subscripts from the pmfs for simplicity, e.g., $P_{X|AC}^n(x^n|A_{ijk}^n, C_{ij}^n)$ will be denoted by $P(x^n|A_{ijk}^n, C_{ij}^n),$ and $Q^n_{XY}(x^n,y^n)$ will be denoted by $Q(x^n,y^n)$, respectively. Let $R\triangleq R_a+R_c+R_o$,  and choose $\epsilon>0$. Consider the derivation for $\mathbb E_\mathsf{C}\big[\mathbb{D}(\check{P}_{X^nY^n}||Q^n_{XY})\big]$ as follows:
 
\fontsize{10}{1}{
	\begin{equation*}\label{eqn_dbl_y}
	\begin{split} 
	&\mathbb E_\mathsf{C}\big[\mathbb{D}(\check{P}_{X^nY^n}||Q^n_{XY})\big]\\
	&=\mathbb E_\mathsf{C}\Bigg[\sum_{x^n,y^n} \Big( \sum_{i,j,k} \dfrac{P(x^n|A^{n}_{ijk},C^{n}_{ij})P(y^n|A^{n}_{ijk},C^{n}_{ij})}{2^{nR}}\Big)
		\log \Bigg( \sum_{i',j',k'} \dfrac{P(x^n|A^{n}_{i'j'k'},C^{n}_{i'j'})P(y^n|A^{n}_{i'j'k'},C^{n}_{i'j'})}{2^{nR}Q(x^n,y^n)}\Bigg)\Bigg] \\
	&\stackrel{(a)}{=}\!\!\sum_{x^n,y^n} \!\!\mathbb E_{A^{n}_{ijk}\!C^{n}_{ij}} \Bigg[\! \Big( \sum_{i,j,k} \!\dfrac{P(x^n|A^{n}_{ijk},C^{n}_{ij})P(y^n\!|A^{n}_{ijk},C^{n}_{ij})}{2^{nR}}\Big) 
		 \mathbb E_{\mathrm{rest}} \Big[\! \log \Big(\!\! \sum_{i',j',k'} \!\!\dfrac{P(x^n|A^{n}_{i'j'k'},\!C^{n}_{i'j'})P(y^n|A^{n}_{i'j'k'},\!C^{n}_{i'j'})}{2^{nR}Q(x^n,y^n)}\Big) \Big| A^{n}_{ijk}C^{n}_{ij}\Big]\!\Bigg] \\ 
	&\stackrel{(b)}{\leq}\!\! \sum_{x^n,y^n} \!\! \mathbb E_{A^{n}_{ijk}\!C^{n}_{ij}}\Bigg[\! \Big( \sum_{i,j,k} \dfrac{P(x^n|A^{n}_{ijk},C^{n}_{ij})P(y^n|A^{n}_{ijk},C^{n}_{ij})}{2^{nR}}\Big)\!
	\log \Big(\mathbb E_{\mathrm{rest}} \Big[\!\!\sum_{i',j',k'}\! \!\dfrac{P(x^n|A^{n}_{i'j'k'},\!C^{n}_{i'j'})P(y^n|A^{n}_{i'j'k'},\!C^{n}_{i'j'})}{2^{nR}Q(x^n,y^n)} \Big| A^{n}_{ijk}C^{n}_{ij}\Big]\! \Big)\!\Bigg]
	\end{split}
	\end{equation*}
    \begin{equation*} \hspace{-5ex}
	\begin{split}  
	& \stackrel{(c)}{=}\! \sum_{x^n,y^n} \sum_{a^{n}_{ijk},c^{n}_{ij}} \sum_{i,j,k} \dfrac{P(x^n,y^n,a^{n}_{ijk},c^{n}_{ij})}{2^{nR}}\log \Bigg( \sum_{\substack{i',j',k':\\(i',j',k')=(i,j,k)}} \mathbb E_{A^{n}_{ijk}\!C^{n}_{ij}} \Big[ \dfrac{P(x^n|A^{n}_{i'j'k'},C^{n}_{i'j'})P(y^n|A^{n}_{i'j'k'},C^{n}_{i'j'})}{2^{nR}Q(x^n,y^n)}\Big|A^{n}_{ijk}C^{n}_{ij}\Big]\\
	&\hspace{2.6in} +\sum_{\substack{i',j',k':\\(i',j')=(i,j),(k'\neq k)}} \mathbb E_{A^{n}_{ijk}\!C^{n}_{ij}} \Big[ \dfrac{P(x^n|A^{n}_{i'j'k'},C^{n}_{i'j'})P(y^n|A^{n}_{i'j'k'},C^{n}_{i'j'})}{2^{nR}Q(x^n,y^n)}\Big|A^{n}_{ijk}C^{n}_{ij}\Big]\\ 
	&\hspace{2.6in} +\sum_{\substack{i',j',k':\\(i',j')\neq(i,j)}} \mathbb E_{A^{n}_{ijk}\!C^{n}_{ij}} \Big[ \dfrac{P(x^n|A^{n}_{i'j'k'},C^{n}_{i'j'})P(y^n|A^{n}_{i'j'k'},C^{n}_{i'j'})}{2^{nR}Q(x^n,y^n)}\Big|A^{n}_{ijk}C^{n}_{ij}\Big] \Bigg) \\
	&\stackrel{(d)}{=} \sum_{x^n,y^n} \sum_{a^{n}_{ijk},c^{n}_{ij}} \sum_{i,j,k} \dfrac{P(x^n,y^n,a^{n}_{ijk},c^{n}_{ij})}{2^{nR}}\log \Bigg(\dfrac{P(x^n,y^n|a^{n}_{ijk},c^{n}_{ij})}{2^{nR}Q(x^n,y^n)} +\sum_{\substack{i',j',k':\\(i',j')=(i,j),(k'\neq k)}} \dfrac{P(x^n,y^n|c^{n}_{ij})}{2^{nR}Q(x^n,y^n)} \\
	&\hspace{2.6in} +\sum_{\substack{i',j',k':\\(i',j')\neq(i,j)}} \dfrac{Q(x^n,y^n)}{2^{nR}Q(x^n,y^n)} \Bigg) \\ 
	&\stackrel{(e)}{\leq}\sum_{x^n,y^n,a^{n}_{ijk},c^{n}_{ij}} P(x^n,y^n,a^{n}_{ijk},c^{n}_{ij}) \log \Bigg (\dfrac {P(x^n,y^n|a^{n}_{ijk},c^{n}_{ij})}{2^{nR}Q(x^n,y^n)} +(2^{nR_a}) \dfrac{P(x^n,y^n|c^{n}_{ij})}{2^{nR}Q(x^n,y^n)} +1 \Bigg) \\ 
	&\stackrel{(f)}{\leq} \Bigg[ \sum_{(x^n,y^n,a^n,c^n)\in {\cal T}_\epsilon^n( P_{XYAC})}\!\! P(x^n,y^n,a^{n},c^{n}) \log \Bigg (\dfrac{2^{-nH(XY|AC)(1-\epsilon)}} {2^{nR}2^{-nH(XY)(1+\epsilon)}} + \dfrac{2^{-nH(XY|C)(1-\epsilon)}}{2^{n(R_o+R_c)}2^{-nH(XY)(1+\epsilon)}} + 1 \Bigg)\Bigg]\\ 
	&\hspace{3.8in} + \mathbb P\big((x^n,y^n,a^{n},c^{n}) \notin {\cal T}_\epsilon^n( P_{XYAC}) \big)\log(2\mu_{XY}^{-n}+1)\\
	&\stackrel{(g)}{\leq}\Bigg[ \sum_{(x^n,y^n,a^n,c^n)\in {\cal T}_\epsilon^n( P_{XYAC})} P(x^n,y^n,a^{n},c^{n}) \log \Bigg(\dfrac{2^{n(I(XY;AC)+\delta(\epsilon))}}{2^{nR}}+ \dfrac{2^{n(I(XY;C)+\delta(\epsilon))}}{2^{n(R_o+R_c)}}+1 \Bigg) \Bigg]\\ 
	&\hspace{3.8in} + \big(2{\cal |X||Y||A||C|}e^{-n\epsilon^2\mu_{XYAC}}\big)\log(2\mu_{XY}^{-n}+1)\\ 
	&\stackrel{(h)}{\leq} \epsilon'.
	\end{split}
	\end{equation*} 
	\vspace*{2pt}
}
\noindent In this argument:
\begin{itemize} 
\item[($a$)] follows from the law of iterated expectations. Note that we have used $(a^{n}_{ijk},c^{n}_{ij})$ to denote the codewords corresponding to the indices $(i,j,k)$, and $(a^{n}_{i'j'k'},c^{n}_{i'j'})$ to denote the codewords corresponding to the indices $(i',j',k')$.
\item[($b$)] follows from Jensen's inequality \cite{EoIT:2006}.
\item[($c$)] follows from dividing the inner summation over the indices $(i',j',k')$ into three subsets based on the indices $(i,j,k)$ from the outer summation.
\item[($d$)] follows from taking the expectation within the subsets in (c) such that when
\begin{itemize}
	\item $(i',j')=(i,j),(k'\neq k)$: $a^{n}_{i'j'k'}$ is conditionally independent of $a^{n}_{ijk}$ following the nature of the codebook construction (i.e., i.i.d.~at random);
	\item $(i',j')\neq(i,j)$: both codewords ($a^{n}_{ijk},c^{n}_{ij}$) are independent of $(a^{n}_{i'j'k'},c^{n}_{i'j'})$ regardless of the value of $k$. As a result, the expected value of the induced distribution with respect to the input codebooks is the desired distribution $Q^n_{XY}$ \cite{cuff2010coordination}.
\end{itemize} 
\item[($e$)] follows from
\begin{itemize}
	\item $(i',j',k')=(i,j,k)$: there is only one pair of codewords $(a^{n}_{ijk},c^{n}_{ij})$;
	\item when $(k'\neq k)$ while $(i',j')=(i,j)$ there are $(2^{nR_a}-1)$ indices in the sum; 
	\item $(i',j')\neq(i,j)$: the number of the indices is at most $2^{nR}.$	
\end{itemize} 
\item[($f$)] results from splitting the outer summation: The first summation contains typical sequences and is bounded by using the probabilities of the typical set. The second summation contains the tuple of sequences when the pair of actions sequences $x^n, y^n$ and codewords $c^n,a^n$ are not $\epsilon$-jointly typical (i.e., $(x^n,y^n,a^n,c^n)\notin {\cal T}_\epsilon^n( P_{XYAC})$). This sum is upper bounded following \cite{BK14} with $\mu_{XY} =\min^{*}_{x,y} \big(P_{XY}(x,y)\big)$.
\item[($g$)] follows from Chernoff bound on the probability that a sequence is not strongly typical \cite{kramer2008TinMUIT} where $\mu_{XYAC} =  \\\min^{*}_{x,y,a,c} ( P_{XYAC}(x,y,a,c) ) $ and $\delta(\epsilon)$ denotes a
	positive function of $\epsilon$ that vanishes as $n$ goes to infinity, i.e.,~$\delta(\epsilon) \rightarrow 0$ as $n \rightarrow \infty$.
\item[($h$)] Consequently, the contribution of typical sequences can be made asymptotically smaller than some $\epsilon'> 0$ if
\[
R_a+R_o+R_c >  I(XY;AC),\quad
R_o+R_c  >  I(XY;C),
\]
while the second term converges to zero exponentially fast with $n$, i.e.,~$\big(2{\cal |X||Y||A||C|}e^{-n\epsilon^2\mu_{XYAC}}\big)\log(2\mu_{XY}^{-n}+1) 
	\xrightarrow{n\rightarrow\infty} 0$ and $\epsilon' \rightarrow 0$ as $n \rightarrow \infty$.
\end{itemize}
Finally, if \eqref{eqn-resolve1} and \eqref{eqn-resolve2} are satisfied, by applying Pinsker's inequality \cite{pinsker1964information} we have  
\begin{align} \label{eq:resolvProof}
	\mathbb E_\mathsf{C}\big[||\check{P}_{X^nY^n}-Q^n_{XY}||_{{\scriptscriptstyle TV}}\big] 
	 & \leq  \mathbb E_\mathsf{C}\Big[\sqrt{2\mathbb{D}(\check{P}_{X^nY^n}||Q^n_{XY})} \;\Big] \notag\\
	& \leq \sqrt{2\mathbb E_\mathsf{C} \big[\mathbb{D}(\check{P}_{X^nY^n}||Q^n_{XY})\big]}\mathop{\longrightarrow}^{n\rightarrow \infty} 0.
\end{align}
\end{proof}

\begin{remark}
Given $\epsilon>0$, $R_a$, $R_o$, $R_c$ satisfying \eqref{eqn-resolve1} and \eqref{eqn-resolve2}, it follows from \eqref{eq:resolvProof} that there exist an $n\in\mathbb N$ and a random codebook realization for which the code-induced pmf between the indices and the pair of actions satisfies
\begin{align}
||\check{P}_{X^nY^n}-Q^n_{XY}||_{{\scriptscriptstyle TV}}  <\epsilon.
\end{align}
\end{remark}

\subsubsection{Decodability constraint} \label{sec:DecoCons}
Since the operation at Node $\mathsf Y$ in Fig.~\ref{fig:StrongCoordinationAllied} involves the decoding of $I$ and thereby the codeword $C^n(I,J)$ using $B^n$ and $J$, the induced distribution of the scheme for the allied problem that is given in \eqref{eqn-codeinducedAct} will not match that of \eqref{eqn-codeinduced} unless and until we ensure that the decoding succeeds with high probability as $n\rightarrow \infty$. The following lemma quantifies the necessary rate for this decoding to succeed asymptotically almost always.

\begin{lemma}[Decodability constraint]
\label{Lma:Decod}
Let $\hat I, C^n_{\hat IJ}$ be the output of a typicality-based decoder that uses common randomness $J$ to decode the index $I$ and the sequence $C^n_{{I}J}$ from $B^n$. Let $\mathbb P[\hat{I}\neq I]$ be the probability that the decoding fails for a realization of the random codebook.
If the rate for the index $I$ satisfies $R_c<I(B;C)$ then, { 
\renewcommand{\theenumi}{\roman{enumi}}
\begin{enumerate}
\item $\mathbb E_\mathsf{C}\big[\mathbb P[\hat{I}\neq I]\big]\rightarrow 0$ as $n\rightarrow \infty$, and \label{itm:proof2Itm1}
\item $\lim\limits_{n\rightarrow \infty} \mathbb E_\mathsf{C}\big[ \norm{\check{P}_{X^nY^nIJK}-\mathring{P}_{X^nY^nIJK}}_{{\scriptscriptstyle TV}}\big] =0.$ \label{itm:proof2Itm2}
\end{enumerate}}
\end{lemma}
\begin{proof}
We start the proof of \textit{\ref{itm:proof2Itm1}{)}} by calculating the average probability of error, averaged over all codewords in the codebook and averaged over all random codebook realizations as follows:
 \begin{align}
 \mathbb E_\mathsf{C}\big[\mathbb P[\hat{I}\neq I]\big]&= \sum_{\mathsf{C}} P_{\mathsf C}(\mathsf c) \mathbb P[\hat{I}\neq I] \notag\\
 &= \sum_{\mathsf{C}} P_{\mathsf C}(\mathsf c) \sum_{i,j,k} \frac{1}{2^{nR}}\mathbb P\Big[\hat{I}\neq I \Big| \substack{I=i\\ J=j\\K=k}\Big] \notag\\
 &=\sum_{i,j,k} \frac{1}{2^{nR}} \sum_{\mathsf{C}} P_{\mathsf C}(\mathsf c) \mathbb P\Big[\hat{I}\neq I \Big| \substack{I=i\\ J=j\\K=k}\Big] \notag\\
 &\stackrel{(a)}{=}\mathbb P\Big[\hat{I}\neq I \Big| \substack{I=1\\ J=1\\K=1}\Big] , \label{eq:DecodProof}
\end{align}
where in ($a$) we have used the fact that the conditional probability of error is independent of the triple of indices due to the i.i.d nature of the codebook construction. Also, due to the random construction and the properties of jointly typical set, we have 
\begin{equation*}
{\mathbb P} \big((A_{111}^n,B^n,C^n_{11})\in {\cal T}_\epsilon^n(P_{ABC})\big) \xrightarrow{n\rightarrow\infty} 1.
\end{equation*}

We now continue the proof by constructing the sets for each $j$ and $ b^n\in\mathcal B^n$ that Node $\mathsf Y$ will construct to identify the transmitted index as 
\begin{equation*}
	\hat{S}_{j,b^n,\mathsf c}\triangleq  \{i: (b^n,c^n_{ij}) \in {\cal T}_\epsilon^n(P_{BC})\}.
\end{equation*}

\noindent The set $\hat{S}_{j,b^n\!,\mathsf c}$ consists of indices $i\in I$ such that for a given common randomness index $J=j$ and channel realization $B^n=b^n$, the sequences $(b^n,c_{ij}^n)$ are jointly-typical. Assuming $(i,j,k)=(1,1,1)$ was realized, and if $\hat{S}_{1,b^n,\mathsf c}=\{1\}$, then the decoding will be successful. The probability of this event occurring is divided into two steps as follows.

\begin{itemize}
	\item First, assuming $(i,j,k)=(1,1,1)$ was realized, for successful decoding, $1$ must be an element of $\hat{S}_{J,B^n,\mathsf c}$. The probability of this event can be bounded as follows:
	\begin{align*} 
	\mathbb E_\mathsf C \Big[\mathbb P\Big[I \in \hat{S}_{J,B^n,\mathsf C}\Big|\substack{I=1\\J=1\\{K=1}}\Big]\Big]
	&=\sum_{a^n,b^n,c^n} \Big(P^n_{C}(c^n) P^n_{A|C}(a^n|c^n) P^n_{B|A}(b^n|a^n) \mathds{1}\big((c^n,b^n)\in {\cal T}_\epsilon^n(P_{BC})\big)\Big)\\
	&\stackrel{}{=}\sum_{b^n,c^n} P^n_{BC}(b^n,c^n)\mathds{1}\big((b^n,c^n)\in {\cal T}_\epsilon^n(P_{BC})\big)\\
	&\stackrel{(a)}{\geq} 1-\delta(\epsilon) \xrightarrow{n\rightarrow\infty} 1,
	\end{align*} 
	
	where (a) follows from the properties of jointly typical sets and $\delta(\epsilon) \rightarrow 0$ as $n \rightarrow \infty$.
	
	\item Next, assuming again that $(i,j,k)=(1,1,1)$ was realized, for successful decoding, no index greater than or equal to $2$ must be an element of $\hat{S}_{J,B^n,\mathsf C}$. The probability of this event can be bounded as follows:
	\begin{align*} 
	\mathbb E_\mathsf C \mathbb P\Big [\hat{S}_{J,B^n,\mathsf C} \cap \{2,\dots,2^{nR_c}\} = \emptyset\Big|\substack{I=1\\J=1\\K=1}\Big]
	&=1- \sum_{i' \neq 1} \mathbb{E}_\mathsf C\mathbb P\Big [i' \in \hat{S}_{J,B^n,\mathsf C}\Big|\substack{I=1\\J=1\\K=1}\Big] \\ 
	&=1- \sum_{i' \neq 1} \mathbb P[(C^n_{i'1},B^n)\in {\cal T}_\epsilon^n(P_{BC})]\\
	&\stackrel{(a)}{\geq}1- \sum_{i' \neq 1} 2^{-n(I(B;C)-\delta(\epsilon))}\\
	&=1-(2^{nR_c}-1) 2^{-n(I(B;C)-\delta(\epsilon))}\\
	&=1- 2^{-n(I(B;C)-R_c-\delta(\epsilon))} + 2^{-nI(B;C)}\\
	&\stackrel{(b)}{\geq} 1-\delta(\epsilon) \xrightarrow{n\rightarrow\infty} 1,
	\end{align*}
	where
	($a$) follows from the packing lemma \cite{elGamal2011NITetwork}, and
	($b$) results if $R_c <I(B;C)-\delta(\epsilon)$ and sufficiently large $n$ yield $\delta(\epsilon)\rightarrow 0$. 		
	\noindent Then from \eqref{eq:DecodProof}, the claim in \textit{\ref{itm:proof2Itm1}{)}} follows as given by		
	\begin{align*}
	\mathbb E_\mathsf{C}\big[\mathbb P[\hat{I}\neq I]\big] & = \mathbb E_\mathsf C\mathbb P\Big[\hat{I}\neq I\Big|\substack{I=1\\J=1\\K=1}\Big] \\
	&\leq \left(\mathbb E_\mathsf C \mathbb P\Big[I \notin \hat{S}_{J,B^n,\mathsf C}\Big|\substack{I=1\\J=1\\{K=1}}\Big]
	+ \mathbb E_\mathsf C \mathbb P\Big [\hat{S}_{J,B^n,\mathsf C} \cap \{2,\dots,2^{nR_c}\}\!\neq\! \emptyset\Big|\substack{I=1\\J=1\\K=1}\Big]\right) 
	\xrightarrow{n\rightarrow\infty} 0.
	\end{align*}
\end{itemize} 
		
Finally, the proof of \textit{\ref{itm:proof2Itm2}{)}} follows in a straightforward manner. If the previous two conditions are met, then $\mathbb E_\mathsf{C}[\mathbb P[\hat{I}\neq I]]\rightarrow 0$ and  $\mathbb E_\mathsf{C} [{P_{\hat{I}|B^nJ}(\hat{i}|b^n,j)]\rightarrow\delta_{I\hat{I}}}$, 
where $\delta_{I\hat{I}}$ denotes the Kronecker delta. Consequently, the claim then follows by simple algebraic manipulation of  \eqref{eqn-codeinducedAct} and \eqref{eqn-codeinduced} as
 \begin{align}\lim\limits_{n\rightarrow \infty}  \mathbb E_\mathsf{C}\big[ \norm{\check{P}_{X^nY^nIJK}-\mathring{P}_{X^nY^nIJK}}_{{\scriptscriptstyle TV}}\big]  =0.\end{align}
\end{proof}

\vspace*{-0.5ex}
\subsubsection{Independence constraint}\label{sec:independence}
We complete modifying the allied structure in Fig.~\ref{fig:StrongCoordinationAllied} to mimic to the original problem
with a final step. By assumption, we have a natural independence between the
action sequence $X^n$ and the common randomness $J$. As a result, the joint
distribution over $X^n$ and $J$ in the original problem is a product of the
marginal distributions $Q^n_{X}$ and $P_J$. To mimic this
behavior in the scheme for the allied problem, we artificially enforce independence by ensuring that the mutual information between $X^n$ and $J$ vanishes. This process is outlined in Lemma~\ref{Lma:Secrecy}.

\begin{lemma}[Independence constraint]
	\label{Lma:Secrecy}
	Consider the scheme for the allied problem given in Fig.~\ref{fig:StrongCoordinationAllied}. Both $	I(J;X^n)\rightarrow 0$ and $\mathbb E_\mathsf{C} \big[||\check{P}_{X^nJ}-Q^n_{X}P_J||_{{\scriptscriptstyle TV}} \big] \rightarrow 0$ as $n\rightarrow \infty$ 
	if the code rates satisfy 
	\begin{align}
    R_a+R_c &> I(X;AC),\label{eqn-Indep1}\\
    R_c &> I(X;C).\label{eqn-Indep2}
    \end{align}
\end{lemma}

The proof of Lemma~\ref{Lma:Secrecy}, shown in Appendix~\ref{Appnx:ProofSecrecyLemma}, builds on the results of
Section~\ref{sec:DecoCons} and the proof of
Lemma~\ref{Lma:Reslv} in Section \ref{sec:ChResCons}, resulting in
\begin{align} \label{eq:indepProof} 
\mathbb E_\mathsf{C} \big[||\check{P}_{X^nJ}-Q^n_{X}P_J||_{\scriptscriptstyle TV}\big]  
&\leq \mathbb E_\mathsf{C} \Big[\sqrt{2\mathbb{D}(\check{P}_{X^nJ}||Q^n_{X}P_J)}\;\Big] \notag\\
&\leq \sqrt{2\mathbb E_\mathsf{C} \big[\mathbb{D}(\check{P}_{X^nJ}||Q^n_{X}P_J)\big]}\mathop{\longrightarrow}^{n\rightarrow \infty} 0.
\end{align}

\begin{remark}
	Given $\epsilon>0$, $R_a$, $R_c$ meeting \eqref{eqn-Indep1} and \eqref{eqn-Indep2}, it follows from \eqref{eq:indepProof} that there exists an $n\in\mathbb N$ and a random codebook realization for which the code-induced pmf between the common randomness $J$ and the actions of Node $\mathsf X$ satisfies
	\begin{align} 
	||\check{P}_{X^nJ}-Q^n_{X}P_J||_{{\scriptscriptstyle TV}}  <\epsilon.
	\end{align}
\end{remark}

In the original problem of Fig.~\ref{fig:P2PCoordination}, the input action sequence $X^n$ and the index $J$ from the common randomness source are available and the $A$- and $C$-codewords are to be selected. Now, to devise a scheme for the strong coordination problem, we proceed as follows. We let Node $\mathsf X$ choose indices $I$ and $K$ (and, consequently, the $A$- and $C$-codewords) from the realized $X^n$ and $J$ using the conditional distribution $\mathring P_{I,K|X^n,J}$. The joint pmf of the actions and the indices is then given by
\begin{align}
\hat{P}_{X^nY^nIJK} \triangleq Q^n_{X} P_J \mathring P_{I,K|X^n,J} \mathring P_{Y^n| I,J,K}. \label{eqn-randomcoordscheme}
\end{align}
As a result, from the allied scheme of Fig.~\ref{fig:StrongCoordinationAllied} we obtain the joint scheme illustrated in Fig.~\ref{fig:StrongCoordination}.
\begin{figure}[h!]
	\centering
	{\includegraphics[scale=0.625]{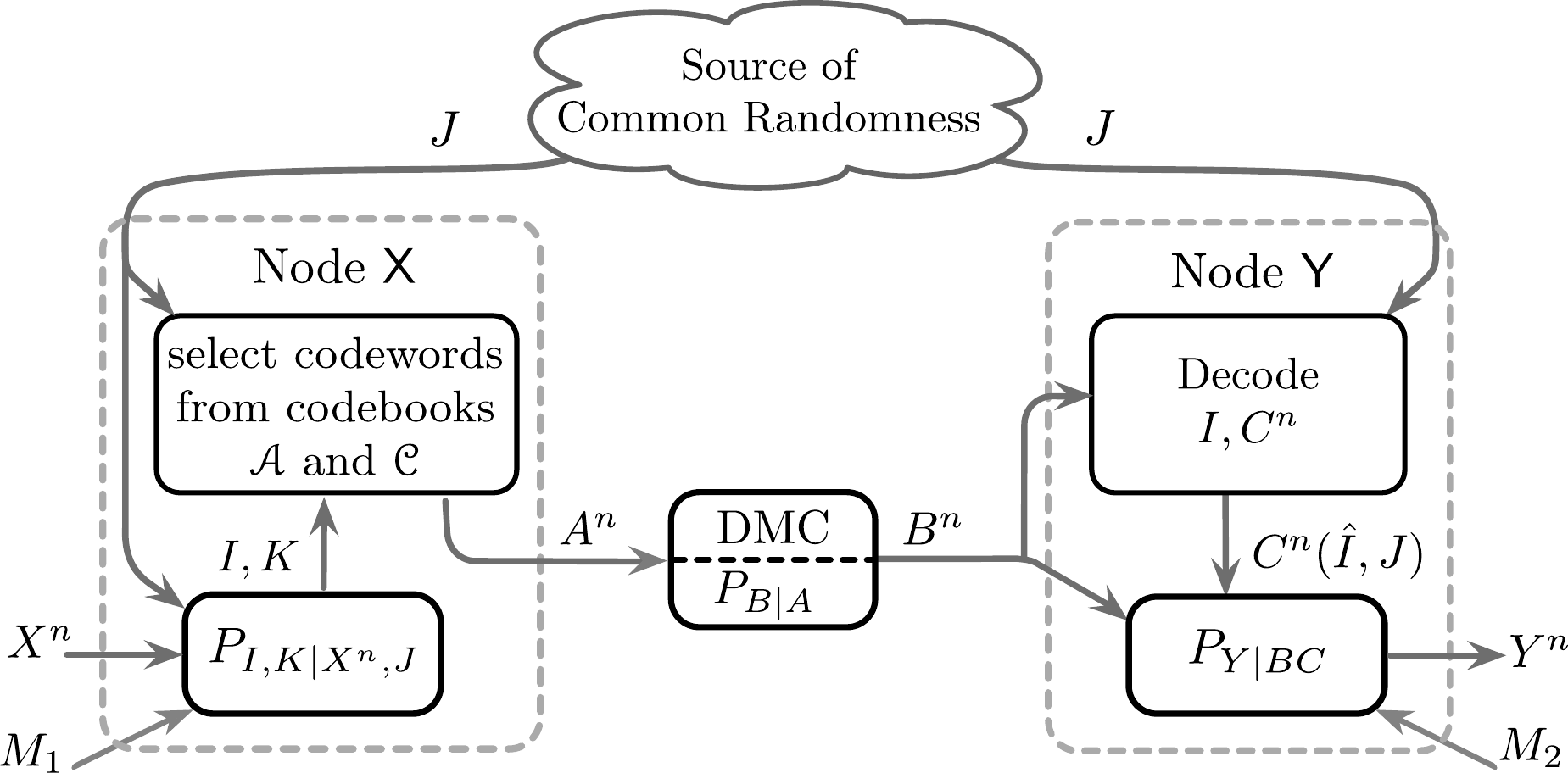}}
	\caption{The joint strong coordination coding scheme.} 
	\label{fig:StrongCoordination}
\end{figure}
\FloatBarrier

\noindent Finally, we can argue that
\begin{align}
\lim_{n\rightarrow \infty} \mathbb E_\mathsf{C} [\norm{\hat{P}_{X^n,Y^n}- Q^n_{XY}}_{{\scriptscriptstyle TV}} ] =0, 
\end{align} 
since the total variation between the marginal pmf $\hat{P}_{X^n,Y^n}$ and the design pmf $Q^n_{XY}$ can be bounded as 
\begin{equation*}
\begin{split} 
\norm{\hat{P}_{X^nY^n}-Q^n_{XY}}_{\scriptscriptstyle TV}
&\stackrel{(a)}{\leq} \norm{\hat{P}_{X^nY^n} - \check{P}_{X^nY^n}}_{\scriptscriptstyle TV}+ \norm{\check{P}_{X^nY^n} -Q^n_{XY}}_{\scriptscriptstyle TV} \notag\\  
&\stackrel{(b)}{\leq} \norm{\hat{P}_{X^nY^nIJK}-\check{P}_{X^nJ}\mathring P_{IKY^n|X^n,J}}_{\scriptscriptstyle TV} +\norm{\check{P}_{X^nY^n}-Q^n_{XY}}_{\scriptscriptstyle TV}\notag\\
&\stackrel{(c)}{=}\norm{Q^n_{X} P_J -\check{P}_{X^nJ}}_{\scriptscriptstyle TV}+ \norm{\check{P}_{X^nY^n} -Q^n_{XY}}_{\scriptscriptstyle TV} 
\end{split}
\end{equation*}

\noindent where\\
(a) follows from the triangle inequality;\\
(b) follows from \eqref{eqn-randomcoordscheme} and \cite[Lemma V.1]{cuff2013distributed};\\ (c) follows from \cite[Lemma V.2]{cuff2013distributed}.\\
Note that the terms on the RHS of the above equation can be made
vanishingly small provided the resolvability, decodability, and independence
conditions are met. Thus, by satisfying the conditions stated in
Lemmas~\ref{Lma:Reslv}-\ref{Lma:Secrecy}, the coordination scheme
defined by \eqref{eqn-randomcoordscheme} achieves strong coordination
asymptotically between Nodes $\mathsf X$ and $\mathsf Y$ by communicating over the DMC $P_{B|A}$. Note that since the operation at Nodes $\mathsf X$ and $\mathsf Y$ amount to index selection according to $\mathring P_{I,K|X^n,J}$, and generation of $Y^n$ using the DMC $P_{Y|BC}$, both operations are randomized. The last step consists in viewing the local randomness as the source of randomness in the operations at Nodes $\mathsf X$ and $\mathsf Y$. This is detailed in the following paragraph.\\

\subsubsection{Local randomness rates}
As seen from Fig.~\ref{fig:StrongCoordination}, at Node $\mathsf X$, local randomness $M_1$ is employed to randomize the selection of indices $(I,K)$ by synthesizing the channel $\mathring P_{IK|X^nJ}$ whereas Node $\mathsf Y$ utilizes its local randomness $M_2$ to generate the action sequence $Y^n$ by simulating the channel $P_{Y|BC}$. Using the list decoding and likelihood arguments of~\cite[Section~IV.B]{VKB16}, we can argue that for any given realizations of $J$, the minimum rate of local randomness required for the probabilistic selection of indices $I,K$ can be derived by quantifying the number of $A$ and $C$ codewords (equally identifying a list of index tuples $(I,K)$) jointly typical with $X^n=x^n.$ Quantifying the list size as in \cite{VKB16} yields $\rho_1 \geq R_a+R_c-I(X;AC)$. At Node $\mathsf Y$,
the necessary local randomness for the generation of the action sequence is bounded by the channel simulation rate of DMC $P_{Y|BC}$~\cite{Steinberg-Verdu-IT-1994}. Thus,  
$\rho_2 \geq H(Y|BC)$. Combining the local randomness rates constraints with the constraints in Lemmas \ref{Lma:Reslv}-\ref{Lma:Secrecy}, we obtain
the following inner bound to the strong coordination capacity region.

\begin{theorem}\label{Thm:JointCRR}
	A tuple $(R_o, \rho_1,\rho_2)$ is achievable 
	for the strong noisy communication setup in
	Fig.~\ref{fig:P2PCoordination} if for some $R_a,R_c\geq0$, there exist
	auxiliary random variables $(C, A)$ jointly correlated with the actions $(X,Y)$
	according to pmf $P_{XYABC}$ such that $P_{XYABC}=P_{AC}P_{X|AC}P_{B|A}P_{Y|BC}$, the marginal distribution $P_{XY} = Q_{XY}$, and 		
	\begin{subequations}\label{equ:JSRR}
		\begin{align} 
		R_a+R_o+R_c & {\;>\;} I(XY;AC), \label{equ:JSRR1}\\ 
		R_o+R_c &{\;>\;} I(XY;C),\label{equ:JSRR2}\\ 
		R_a+R_c &{\;>\;} I(X;AC),\label{equ:JSRR3}\\
		R_c &{\;>\;} I(X;C),\label{equ:JSRR4}\\
		R_c &<I(B;C),\label{equ:JSRR5}\\
		\rho_1 &{\;>\;} R_a+R_c-I(X;AC),\label{equ:JSRR6}\\ 
		\rho_2 &{\;>\;} H(Y|BC).\label{equ:JSRR7}  
		\end{align}
	\end{subequations}
\end{theorem}

\subsection{Capacity Region for Special Cases}
Here we characterize two special cases of the proposed joint coordination-channel scheme in which the scheme is optimal. However, the optimality of the joint scheme in general is still open.
	
\subsubsection{ Node $\mathsf Y$ actions are determined by Node $\mathsf X$}
One straightforward observation is when the action at Node $\mathsf Y$ is a (deterministic) function of the action supplied to Node $\mathsf X$ by nature, i.e.,~$H(Y|X)=0$. For this case, the setup of strong coordination over a noisy channel reduces to a pure lossless compression and communication problem where both local and common randomness are not required. We now verify that this is a special case of Theorem~\ref{Thm:JointCRR} for which equality in the rate conditions is actually obtained.

\begin{proposition} There exists an achievable joint coordination-channel coding scheme for the strong coordination setup in Fig.~\ref{fig:P2PCoordination}, when the action of Node $\mathsf Y$ is a deterministic function of the action supplied to Node ${\mathsf X}$, such that (\ref{eq:StrngCoorCondtion}) is satisfied if
		\begin{align}
		H(Y) & <  \mathbb{C}_{P_{B|A}}.\label{equ:JSSC1}
		\end{align}
 Here, $\mathbb{C}_{P_{B|A}}$ is the channel capacity for the channel $P_{B|A}$ defined as $\mathbb{C}_{P_{B|A}}\triangleq\max_{P_A}I(A;B).$
 
 Conversely, for a given $\epsilon>0$, any $\epsilon$-code achieving the
 strong coordination rate tuple ($R_o,\rho_1,\rho_2$), i.e., a (strong coordination) code that induces the joint distribution $\hat{P}_{X^nY^n}$ such that $\norm{\hat{P}_{X^nY^n}-Q^{n}_{XY}}_{{\scriptscriptstyle TV}} <\epsilon$, must satisfy \eqref{equ:JSSC1}.

\end{proposition}

\begin{proof}
We begin by verifying the achievability part with the following choice for auxiliary random variables: $A$ is independent of $Y$ and $X$, where $P_A$ is the capacity achieving input distribution for the channel $P_{B|A}$, and $C= (A,Y)$. This selection will result in the joint distribution of Theorem \ref{Thm:JointCRR} taking the form $P_{C|XY}P_{A}P_{B|A}P_{XY}$. Now we proceed with a direct application of the selected auxiliary random variables in Theorem~\ref{Thm:JointCRR} as follows

From \eqref{equ:JSRR1} and \eqref{equ:JSRR2} we have:
\begin{subequations}\label{equ:SC}
\begin{align} 
	R_a+R_o+R_c &> I(XY;AC) = I (XY;C)= H(Y) \label{equ:SC1} \\
	R_o+R_c  &>I(XY;C) =H(Y)  \label{equ:SC2}
\end{align}

Furthermore, from \eqref{equ:JSRR3}-\eqref{equ:JSRR7} we have: 
\begin{align}
R_a+R_c & > I(X;AC)\stackrel{(a)}{=}H(Y) \label{equ:SC3}\\
R_c & > I(X;C) \stackrel{(b)}{=} H(Y) \label{equ:SC4}\\
R_c & < I(B;C)\stackrel{(c)}{=}I(B;A)=\mathbb{C}_{P_{B|A}}\\
\rho_1 & > R_a+R_c-I(X;AC) \stackrel{(d)}{=} R_a+R_c-H(Y)\\ 
\rho_2 & > H(Y|BC) \stackrel{(e)}{=} 0\ 
\end{align}
\end{subequations}
\noindent where 
\begin{itemize} 
	\item[($a$)] follows as a result of $H(Y|X)=0$, the chain rule of mutual information, and \eqref{equ:SC1}, i.e.,~$I(XY;AC)= I(X;AC)$;
  	\item[($b$)] follows from $H(Y|X)=0$, the chain rule of mutual information and \eqref{equ:SC2}, i.e.,~$I(XY;C)= I(X;C)$;
  	\item[($c$)] is a result of selecting $C=(A,Y)$, $A$ being independent of $Y$, and $P_A$ to be the capacity achieving input distribution; 
  	\item[($d$)] follows from ($a$); and
  	\item[($e$)] follows from the fact that $H(Y|BC)=0$ as a direct result of selecting $C=(A,Y).$ 
\end{itemize}
Finally, we can see from \eqref{equ:SC4} that \eqref{equ:SC1}-\eqref{equ:SC3} are redundant and $R_o\geq 0$, $R_a\geq 0$. Similarly,  $\rho_2\geq0$ and by selecting $R_c$ to be arbitrary close to $H(Y)$, the local randomness rate $\rho_1$ is not required to achieve strong coordination in this case. As a result, the rate constraints of Theorem~\ref{Thm:JointCRR} are reduced to 
\begin{align*}
&H(Y)< R_c<\mathbb{C}_{P_{B|A}}. 
\end{align*}

Now for the converse part, assume that the rate tuple ($R_o,\rho_1,\rho_2$) is achievable by an $\epsilon$-code such that for a given $\epsilon>0$ $$\norm{\hat{P}_{X^nY^n}-Q^{n}_{XY}}_{{\scriptscriptstyle TV}} <\epsilon.$$

Let $T$ be is a time-sharing random variable uniformly distributed on the set $\{1,\dots, n\}$ and independent of the induced joint distribution. Then, it follows that 
\begin{align*}
nR_c &\geq  H(I) \geq H(I|J) \\
& \geq I(X^n;I|J)= I(X^n;IJ)\\
&=\sum_{t=1}^{n} I(X_t;IJ|X^{t-1})\\
&\stackrel{(a)}{=}\sum_{t=1}^{n} I(X_t;IJX^{t-1})\\
&\stackrel{(b)}{=}\sum_{t=1}^{n} I(X_t;C_t)\\
&= n I(X_T;C_T|T)\\
&\stackrel{(c)}{=}nI(X_T;C_T T)\\
&\stackrel{(d)}{=}nI(X;C)\\
&\stackrel{(e)}{=}n I(X;Y)\\
&\stackrel{(f)}{=} nH(Y)
\end{align*}
\noindent where
\begin{enumerate}
	\item[($a$)] because $X^n$ is i.i.d;
	\item[($b$)] follows by defining an auxiliary random variable $C_t\triangleq(I,J,X^{t-1})$;
	\item[($c$)] holds due to the fact that $X_T$ is independent of $T;$
	\item[($d$)] follows by defining $C\triangleq(C_T,T)$, and since $X^n$ is i.i.d.~we have $X_T=X;$ 
	\item[] \hspace{-3.5ex}($e$)-($f$) follows from the fact that $Y$ is a deterministic mapping of $C$ as a direct result of $C=(A,Y)$, that $A$ is independent of $Y,$ and $H(Y|X)=0$.
\end{enumerate}
Lastly, dividing by $n$, we obtain 
\begin{equation}\label{eq:conv1}
R_c \geq H(Y).
\end{equation} 

Now consider the argument
\begin{align*}
nR_c &= H(I)= H(I|\hat{I})+ I(I;\hat{I})\\
&\stackrel{(a)}{\leq}  1+ \mathbb P_e^{(n)}nR_c+ I(I;\hat{I})\\
&\stackrel{(b)} {\leq}  1+ \mathbb P_e^{(n)}nR_c+ I(IX^n;JB^n)\\
&= 1+ \mathbb P_e^{(n)}nR_c+ I(IJX^n;B^n)+H(J|B^n)-H(J|I)\\
&= 1+ \mathbb P_e^{(n)}nR_c+ I(C^n;B^n)+H(J|B^n)-H(J|I)\\
&\leq 1+ \mathbb P_e^{(n)}nR_c+H(J|B^n)-H(J|I)+\sum_{i=1}^{n}I(C_i;B_i)\\
&\stackrel{(c)}{=} 1+ \mathbb P_e^{(n)}nR_c+\sum_{i=1}^{n}I(A_i;B_i)\\
&\stackrel{(d)}{\leq} 1+ \mathbb P_e^{(n)}nR_c+n \mathbb{C}_{P_{B|A}}
\end{align*}
where
\begin{enumerate}
	\item[($a$)] holds by Fano's inequality where $\mathbb P_e^{(n)}\triangleq \mathbb P[\hat{I}\neq I];$
	\item[($b$)] at the decoder $H(\hat{I}|B^n,J)=0$;
	\item[($c$)] by letting $J\triangleq$ constant, i.e.,~$H(J|B^n)=H(J|I)=0$; and by selecting $C=(A,Y)$; and
	\item[($d$)] follows by the definition of the channel capacity for the channel $P_{B|A}$.
\end{enumerate}
Dividing by $n$, we obtain
\begin{align*}
R_c &\leq  \frac{1}{n}+ \mathbb P_e^{(n)}R_c+ \mathbb{C}_{P_{B|A}},
\end{align*}
Finally, by letting $n\rightarrow \infty$, we have $\mathbb P_e^{(n)}\rightarrow 0$ and hence
\begin{equation} \label{eq:conv2}
R_c \leq \mathbb{C}_{P_{B|A}}.
\end{equation}
By \eqref{eq:conv1} and \eqref{eq:conv2} the proof is complete.
\end{proof}

\subsubsection{Deterministic channel}
This is the case when the channel output $B$ is a deterministic function of the channel input $A$ i.e.,~$H(B|A)=0$. Although this special case is discussed in the context of simulating a DMC channel over a deterministic channel \cite{HYBGA17}, we present this case in the context of our achievable construction with rates as stated in Theorem~\ref{Thm:JointCRR}. 

In this case, we select the auxiliary random variables $A$ and $C$ as follows. Let $C=(U,A)$ and select $A$ independent of $X,Y,$ and $U$ with $P_A$ to be the capacity achieving input distribution of the channel $P_{B|A},$ i.e.,~$\mathbb{C}_{P_{B|A}}=\max_{P_A} H(B)$. Let $U$ be an auxiliary random variable related to $X$ and $Y$ via the Markov chain $X-U-Y$. As a result, the joint distribution of Theorem \ref{Thm:JointCRR} takes the form $P_{U}P_{A}P_{B|A}P_{X|U}P_{Y|U}$ and the problem reduces to a two-terminal strong coordination over a noiseless channel and a separate channel coding problem. Accordingly, from Theorem \ref{Thm:JointCRR}, the following rates are achievable. 
\begin{subequations}\label{equ:CRSC2}
\begin{align}
R_o+R_c \stackrel{(a)}{\geq} & I(XY;U) \label{equ:SC5}\\
R_c \stackrel{(b)}{\geq} & I(X;U) \label{equ:SC6}\\
R_c \stackrel{(c)}{\leq} & I(B;A)= H(B) =\mathbb{C}_{P_{B|A}} \label{equ:SC7}\\
\rho_1 \stackrel{(d)}{\geq}& R_c-I(X;U) \label{equ:SC8}\\ 
\rho_2 \stackrel{(e)}{\leq} & H(Y|U) \label{equ:SC9}
\end{align}
\end{subequations}

\noindent where ($a$)-($e$) follows from the choice of $A$ and $C=(U,A)$; (c) follows from the fact that the channel $P_{B|A}$ is deterministic and from the selection of $P_{A}$ to be the capacity achieving input distribution.

Now, the optimality of \eqref{equ:SC5}-\eqref{equ:SC9} follows in a straightforward way from separate channel coding for \eqref{equ:SC7} and strong coordination over noise-free channels \cite[Theorem~3]{VKB16} for the special case of a single hop.


\section{Separate Coordination-Channel Coding Scheme with Randomness Extraction}\label{sec:SepCoorRE}

As a basis for comparison, we will now introduce a separation-based scheme that involves randomness extraction. We first use a $(2^{nR_c},2^{nR_o},n)$ noiseless coordination code with the codebook $\mathscr{U}$ to a generate message $I$ of rate $R_c$. Such a code exists if and only if the rates $R_o, R_c$ satisfy \cite{cuff2010coordination}
\begin{equation}
\begin{aligned}
R_c+R_o &\geq I(XY;U),\\
R_c &\geq I(X;U).\notag
\end{aligned}
\end{equation}

This coordination message $I$ is then communicated over the noisy channel using
a rate-$R_a$ channel code over $m$ channel uses with codebook $\mathscr{A}$. Hence, $R_c= \lambda
R_a$, where $\lambda=m/n$. The probability of decoding error can be made vanishingly small if $R_a <I(A;B)$. Then,
from the decoder output $\hat{I}$ and the common randomness message $J$ we
reconstruct the coordination sequence $U^n$ and pass it though a test
channel $P(Y|U)$ to generate the action sequence at Node $\mathsf Y$. Note that this
separation scheme is constructed as a special case of the joint coordination-channel scheme in Fig.~\ref{fig:StrongCoordination} by choosing $C=U$
and $P_{AC} =P_A P_U$.

In the following, we restrict ourselves to additive-noise DMCs,
i.e.,~\begin{equation} \label{equ:additiveNoise}
B^m=A^m(I)+Z^m,
\end{equation} where $Z$ is the noise random variable drawn from
some finite field $\mathcal{Z}$, and ``$+$'' is the native addition operation in the field. To extract randomness, we exploit the additive nature of the channel to recover the realization of the channel noise from the decoded codeword. Thus, at the channel decoder output we obtain \begin{equation} \label{equ:NoiseSource}
\hat{Z}^m=B^m+A^m(\hat{I}),
\end{equation} where $B^m$ is the channel output and $A^m(\hat{I})$ the corresponding decoded channel codeword. We can then utilize a randomness extractor on $\hat{Z}^m$ to
supplement the local randomness available at Node $\mathsf Y$. The following lemma provides some guarantees with respect to the randomness extraction stage.  

\begin{lemma}
	\label{lem:rand}
	Consider the separation based scheme over a finite-field additive DMC. If $R_a<I(A;B)$ and we let ${m,n\rightarrow \infty}$ with $\frac{m}{n} = \lambda$, the following holds:
	\renewcommand{\theenumi}{\roman{enumi}} {
		\begin{enumerate}
			\item ${\mathbb P [ Z^m \neq \hat Z^m ] \rightarrow 0},$ \label{itm:proofItm1}
			\item ${\frac{1}{m}H(\hat Z^m) \rightarrow H(Z)},$ and \label{itm:proofItm2}
			\item ${I(\hat Z^m; I, \hat I) \rightarrow 0}$. \label{itm:proofItm3}
	\end{enumerate} }
\end{lemma}
\begin{proof}
	Let $P_e$ be the probability of decoding error (i.e.,~$P_{I_e}=\mathbb P[I\neq \hat{I}]$ and $P_{Z_e}=\mathbb P[Z^m \neq
	\hat{Z}^m]$).
	We first show the claim in~\ref{itm:proofItm1}{)}. From the channel coding theorem we obtain that
	$P_{I_e} \leq 2^{-n\epsilon'}$ for some $\epsilon'> 0$. Consequently, from~\eqref{equ:additiveNoise} and~\eqref{equ:NoiseSource} $\mathbb P[Z^m \neq
	\hat{Z}^m]$ will follow directly as $P_{Z_e} \leq
	2^{-m\epsilon'}$.
	
	\noindent Then, the claim in~\ref{itm:proofItm2}{)} is shown as follows: 
	\begin{equation*}
	\begin{split}
	H(\hat{Z}^m)& \stackrel{(a)}{\leq} H(Z^m)+ H(\hat{Z}^m|Z^m)\\
	&\stackrel{(b)}{\leq} mH(Z) + h_2(P_{Z_e})+ P_{Z_e}m\log|{\cal Z}|
	\end{split}
	\end{equation*}
	\begin{equation*}
	\begin{split}
	{\textstyle \frac{1}{m}}H(\hat{Z}^m)&\leq H(Z) + {\textstyle \frac{1}{m}} h_2(P_{Z_e})+ P_{Z_e}\log|{\cal Z}|\\
	{\textstyle \frac{1}{m}}H(\hat{Z}^m)& \xrightarrow{P_{Z_e}\rightarrow 0} H(Z) 				
	\end{split}
	\end{equation*}
	where\\
	($a$) follows from the chain rule of entropy;\\
	($b$) follows from Fano's inequality and the fact that $Z^m \sim  \prod_{l=1}^{m} P_{Z}(\cdot) $;\\ 
	\noindent Finally, the claim in~\ref{itm:proofItm3}{)} is shown by the following chain of inequalities:
	\begin{equation*}\label{equ:SepSchRE}
	\begin{split}
	I(\hat{Z}^m;I,\hat{I}) & \leq I(Z^m,\hat{Z}^m;I,\hat{I})\\
	& \leq I(Z^m, \hat{Z}^m;I)+H(\hat{I}|I)\\
	& = H(\hat{Z}^m|Z^m)-H(\hat{Z}^m|Z^m,I)+ H(\hat{I}|I)\\
	& \leq H(\hat{Z}^m|Z^m)+ H(\hat{I}|I)\\
	& \stackrel{(a)}{\leq} h_2(P_{Z_e})+ P_{Z_e}m\log|{\cal Z}| + h_2(P_{I_e})+P_{I_e}nR_c\\
	& \stackrel{(b)}{\leq}\epsilon 
	\end{split}
	\end{equation*}
	where
	($a$) follows from Fano's inequality; ($b$) follows from 
	$P_{I_e} \leq 2^{-n\epsilon'}$, $P_{Z_e} \leq 2^{-m\epsilon'}$ and $\epsilon,\epsilon'\rightarrow 0$ as $n,m\rightarrow \infty$ respectively.	
\end{proof}

Now, similar to the joint scheme, we can quantify the local randomness at both nodes \cite{VKB16}, and set $\lambda=1$ to facilitate a comparison with the joint scheme from Section \ref{sec:JointScheme}. The following theorem then describes an inner bound to the strong coordination region using the separate-based scheme with randomness extraction. 
\begin{theorem}\label{Thm:SepSchCRR}
	There exists an achievable separation based coordination-channel coding scheme for the strong setup in Fig~\ref{fig:P2PCoordination} such that (\ref{eq:StrngCoorCondtion}) is satisfied if 
	\begin{subequations}\label{equ:SSRR}
		\begin{align}
		R_c+R_o &\geq I(XY;U), \label{equ:SSRR1}\\ 
		R_c &\geq I(X;U),\label{equ:SSRR2}\\
		R_c &<I(A;B),\label{equ:SSRR3}\\ 
		\rho_1 &\geq R_c-I(X;U), \label{equ:SSRR4}\\ 
		\rho_2 &\geq \max \big(0,H(Y|U)-H(Z)\big). \label{equ:SSRR5} 
		\end{align}
	\end{subequations}	
\end{theorem}

The proof follows in a straightforward way from the proof of
Theorem~\ref{Thm:JointCRR} and Lemma~\ref{lem:rand} and is omitted. 

\section{Nested Polar Code for Strong Coordination over Noisy Channels} \label{sec:PolarCode}

Since the proposed joint coordination-channel coding scheme, displayed in Fig.~\ref{fig:StrongCoordination}, is based
on a channel resolvability framework, we adopt a channel resolvability-based polar construction for noise-free strong coordination \cite{chou2016empirical} in combination with polar coding for the degraded broadcast channel \cite{PC_BCGoela2015}.
We now propose a scheme based on polar coding that achieves the inner bound stated in Theorem~\ref{Thm:JointCRR}.

\subsection{Coding Scheme}\label{sec:PolarCodingScheme}
Consider the random variables $X,Y,A,B,C,\widehat{C}$ distributed according to
$Q_{XYABC\widehat{C}}$ over
$\mathcal{X}\times\mathcal{Y}\times\mathcal{A}\times\mathcal{B}\times\mathcal{C}$
such that $X-(A,C)-(B,\widehat{C})-Y$ forms a Markov chain. Assume that $|A|=2$ and the
target joint distribution over the actions $X$ and $Y$, $Q_{XY}$, is achievable with $|C|=2$\footnote[$\dagger$]{For the
	sake of exposition, we only focus on the set of joint
	distributions over $\mathcal{X}\times\mathcal{Y}$ that are
	achievable with binary auxiliary random variables $C, A$, and over
	a binary-input DMC. The scheme can be generalized to non-binary $C, A$ with non-binary polar codes in a straightforward way~\cite{STA09}.}. Let $N\triangleq2^n, n\in \mathbb{N}$. We describe the polar coding scheme as follows.  

\begin{figure}[h]
	\centering
	{\includegraphics[scale=0.87]{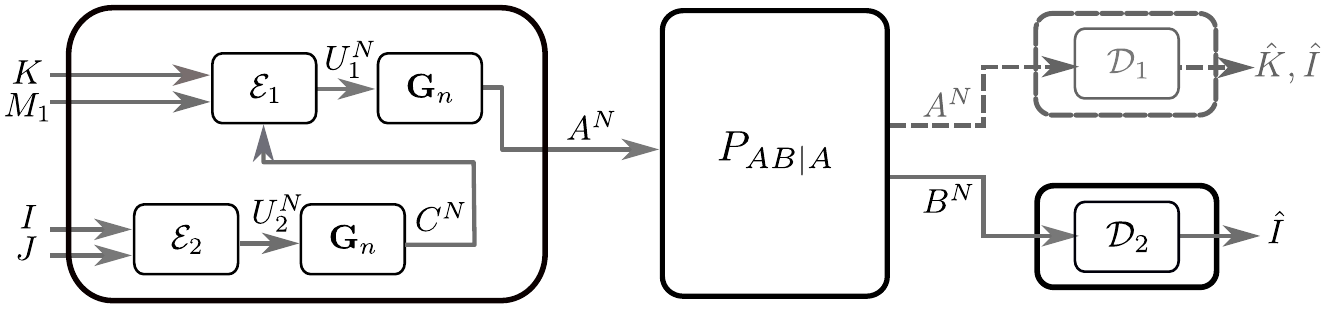}}
	\caption{Block diagram of the superposition polar code.}
	\label{fig:BroadcastChannelPC}
\end{figure}

Consider a 2-user \emph{physically} degraded discrete memoryless
broadcast channel (DM-BC) $P_{AB|A}$ in
Fig.~\ref{fig:BroadcastChannelPC} where $A$ denotes the channel
input and $A,B$ denote the output to the first and second receiver,
respectively. In particular, the channel DMC $P_{B|A}$ is physically
degraded with respect to the \emph{perfect} channel $P_{A|A}$
(we denote this as~$P_{A|A} \succ P_{B|A}$).
We construct the nested polar coding scheme in a similar fashion as
in \cite{PC_BCGoela2015} as this mimics the nesting of the codebooks
$\mathscr{C}$ and $\mathscr{A}$ in Step {i)} of the random coding
construction in~Section~\ref{sec:JointSchemeAchievability}. Here, the second
(weaker) user is able to recover an estimate $\hat{I}$ for its intended message $I$, while the
first (stronger) user is able to recover estimates $\hat{K}, \hat{I}$ for both messages $K$ and
$I$, respectively. Let $C$ be the auxiliary random variable (cloud center)
required for superposition coding over the DM-BC leading to the
Markov chain $C-A-(A,B)$. As a result, the channel $P_{B|C}$ is also degraded with respect to $P_{A|C}$ (i.e.,~$P_{A|C} \succ P_{B|C}$) \cite[Lemma 3]{PC_BCGoela2015}. Note that we let $\widehat{C}$ be the random variable resulting from recovering $C^N$ at Node $\mathsf Y$ from $\hat{I}$ and the shared randomness message $J$. Let ${\bf V}$ be a matrix of the selected codewords $A^N$ and $C^N$ as
\begin{equation}\label{eq:codewordsMtrx}
{\bf V}\triangleq \begin{bmatrix}
A^{N}  \\
C^{N}          
\end{bmatrix}. 
\end{equation}
Now, apply the polar linear transformation ${\bf G}_n,$ where ${\bf G}_n$ is defined in Section~\ref{sec:notation}, as 
\begin{equation}\label{eq:PolztnMtrx}
{\bf U} \triangleq \begin{bmatrix}
U_{1}^{N}  \\
U_{2}^{N}          
\end{bmatrix} = {\bf V} {\bf{G}}_{n}, 
\end{equation}

\noindent where the joint distribution of the random variables in ${\bf U}$ is given by $Q^N_{U_1 U_2}(u_1^N, u_2^N)= Q^N_{AC}(u_1^N{\bf G}_n, u_2^N{\bf G}_n)$. First, consider $C^{N} \triangleq U_{2}^{N} {\bf G}_n$ from \eqref{eq:codewordsMtrx} and \eqref{eq:PolztnMtrx} where $U_{2}^{N}$ is generated by the second encoder ${\cal E}_2$ in Fig.~\ref{fig:BroadcastChannelPC}. For $\beta<\frac{1}{2}$ and  ${\delta_N\triangleq 2^{-N^\beta}}$ we define the very high and high entropy sets
\begin{subequations}\label{eq:C_Sets}
\hspace*{-0.75em}\begin{align}	
{\cal V}_{C} &\triangleq \{i\in \llbracket 1,N \rrbracket    : H(U_{2,i}|U_{2}^{i-1})>1-\delta_{N}\}, \label{eq:C_Sets1}\\
{\cal V}_{C|X} &\triangleq \{i\in \llbracket 1,N \rrbracket  : H(U_{2,i}|U_{2}^{i-1}X^N)>1-\delta_{N}\} \subseteq {\cal V}_{C}, \label{eq:C_Sets2}\\
{\cal V}_{C|XY} &\triangleq \{i\in \llbracket 1,N \rrbracket : H(U_{2,i}|U_{2}^{i-1}X^NY^N)>1-\delta_{N}\}  \subseteq {\cal V}_{C|X}, \label{eq:C_Sets3}\\
{\cal H}_{C|B} &\triangleq \{i\in \llbracket 1,N \rrbracket  : H(U_{2,i}|U_{2}^{i-1}B^N)>\delta_{N}\},\label{eq:C_Sets4}\\
{\cal H}_{C|A} &\triangleq \{i\in \llbracket 1,N \rrbracket  : H(U_{2,i}|U_{2}^{i-1}A^N)>\delta_{N}\}, \label{eq:C_Sets5} 
\end{align}
\end{subequations}

\noindent which by \cite[Lemma 7]{PC_BCRameBloch16} satisfy
\begin{align*}
\lim_{N\rightarrow \infty} \frac{|{\cal V}_{C}|}{N} &= H(C),\; & \lim_{N\rightarrow \infty} \frac{|{\cal V}_{C|X}|}{N} &= H(C|X),\\
\lim_{N\rightarrow \infty} \frac{|{\cal V}_{C|XY}|}{N} &=H(C|XY),\; &\lim_{N\rightarrow \infty} \frac{|{\cal H}_{C|B}|}{N} &=H(C|B),\\
\lim_{N\rightarrow \infty} \frac{|{\cal H}_{C|A}|}{N} &=H(C|A).
\end{align*}

These sets are illustrated in Fig.~\ref{fig:CoordiationSets1}. Note
that the set ${\cal H}_{C|B}$ (exemplary denoted in red in Fig.~\ref{fig:CoordiationSets1}) indicates the noisy bits of the DMC
$P_{B|C}$ (i.e., the unrecoverable bits of the codeword $C^N$ intended
for the weaker user in the DM-BC setup in Fig.~3) and is in general not aligned with other sets. Let
\begin{align*}
{\cal L}_1& \triangleq {\cal V}_{C}\setminus {\cal H}_{C|A} , &&{\cal L}_2\triangleq{\cal V}_{C}\setminus {\cal H}_{C|B},\vspace*{-0.5ex}
\end{align*}

\begin{figure}[h]
	\centering
	{\includegraphics[scale=0.665]{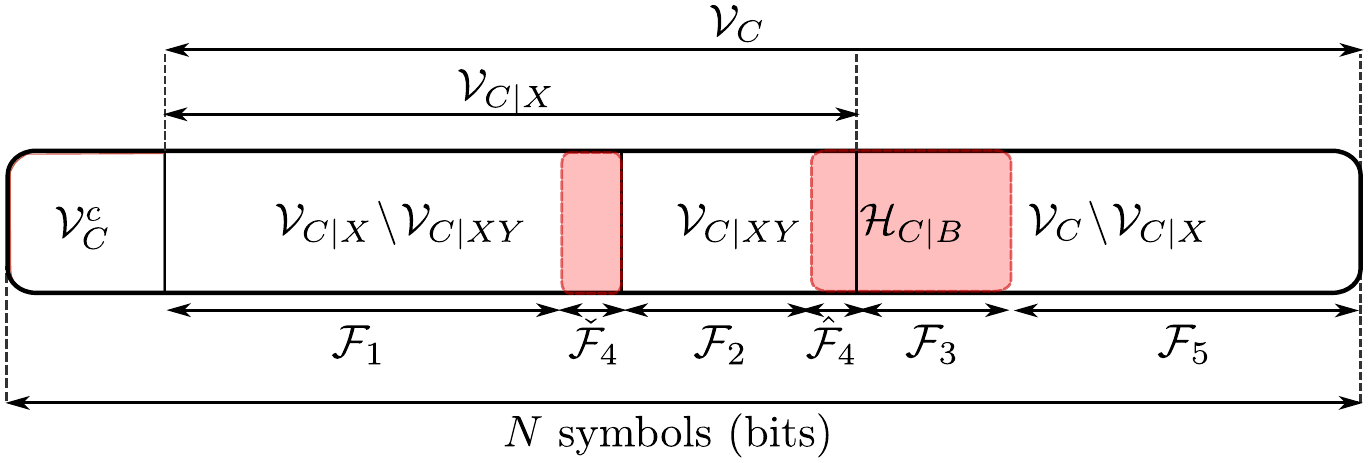}}
	\caption{Index sets for codeword $C$.}
	\label{fig:CoordiationSets1}
\end{figure}

\noindent where the set ${\cal H}_{C|A}$ indicates the noisy bits of the DMC
$P_{A|C}$ (i.e., the unrecoverable bits of the codeword $C^N$  intended
for the stronger user). From the relation $P_{A|C} \succ P_{B|C}$ we
obtain ${\cal H}_{C|A} \subseteq {\cal H}_{C|B}$ and ${\cal
  H}_{C|B}^{c} \subseteq {\cal H}_{C|A}^{c} $, respectively. This
ensures that the polarization indices are guaranteed to be aligned
(i.e.,~ $ {\cal L}_2 \subseteq {\cal L}_1$)
\cite{PC_BCMondelli2014},\cite[Lemma 4]{PC_BCGoela2015}. As a
consequence, the bits decodable by the weaker user are also decodable by the stronger user. 

Accordingly, in terms of the polarization sets in \eqref{eq:C_Sets1}-\eqref{eq:C_Sets4} we define the sets combining channel resolvability for strong coordination and broadcast channel construction as
\begin{align*}
{\cal F}_1&\triangleq ({\cal V}_{C|X}\setminus {\cal V}_{C|XY})\cap {\cal H}_{C|B}^{c},&\hat{\cal F}_4 &\triangleq {\cal H}_{C|BXY}, \\
{\cal F}_2 &\triangleq {\cal V}_{C|XY} \cap {\cal H}_{C|B}^{c},& \check{\cal F}_4 &\triangleq {\cal H}_{C|BX}\!\setminus\!{\cal H}_{C|BXY},\\
{\cal F}_3 &\triangleq  {\cal V}_{C|X}^c \cap {\cal H}_{C|B} ={\cal H}_{C|B}\setminus {\cal H}_{C|BX},& {\cal F}_5&\triangleq ({\cal V}_{C}\setminus {\cal V}_{C|X})\cap {\cal H}_{C|B}^{c} \\
{\cal F}_4&\triangleq {\cal V}_{C|X} \cap {\cal H}_{C|B} = {\cal H}_{C|BX}.
\end{align*}

Now, consider $A^{N}\triangleq U_{1}^{N} {\bf G}_n $ (see
\eqref{eq:codewordsMtrx} and \eqref{eq:PolztnMtrx}), where
$U_{1}^{N}$ is generated by the first encoder ${\cal E}_1$ with
$C^N$ as a side information as seen in
Fig.~\ref{fig:BroadcastChannelPC}. We define the very high entropy
sets illustrated in Fig.~\ref{fig:CoordiationSets2} as
\begin{subequations}\label{eq:A_Sets}
\begin{align}
{\cal V}_{A} &\triangleq \{i\in \llbracket 1,N \rrbracket : H(U_{1,i}|U^{i-1}_1)\!>\! 1\!-\!\delta_N\!\},\label{eq:A_Sets1}\\
{\cal V}_{A|C} &\triangleq \{i\in \llbracket 1,N \rrbracket : H(U_{1,i}|U_{1}^{i-1}C^N)\!>\! 1\!-\!\delta_N \!\}\subseteq {\cal V}_{A}, \label{eq:A_Sets2}\\
{\cal V}_{A|CX} &\triangleq \{i\in \llbracket 1,N \rrbracket : H(U_{1,i}|U_{1}^{i-1}C^NX^N)\!>\! 1\!-\!\delta_N\!\} \subseteq {\cal V}_{A|C},\label{eq:A_Sets3}\\
{\cal V}_{A|CXY} &\triangleq \{i\in \llbracket 1,N \rrbracket : H(U_{1,i}|U_{1}^{i-1}C^NX^NY^N)\!>\! 1\!-\!\delta_N\!\} \subseteq {\cal V}_{A|CX},\label{eq:A_Sets4}
\end{align}
\end{subequations}
\noindent	
satisfying
\begin{align*}
\lim_{N\rightarrow \infty} \frac{|{\cal V}_{A}|}{N} &= H(A),& \lim_{N\rightarrow \infty} \frac{|{\cal V}_{A|CX}|}{N} &=H(A|CX),\\
\lim_{N\rightarrow \infty} \frac{|{\cal V}_{A|C}|}{N} &= H(A|C), &\lim_{N\rightarrow \infty} \frac{|{\cal V}_{A|CXY}|}{N} &=H(A|CXY).\\
\end{align*}

\noindent Note that, in contrast to Fig.~\ref{fig:CoordiationSets1}, here
there is no channel dependent set overlapping with all other sets as
$P_{A|A}$ is a noiseless  channel with rate $H(A)$ and hence ${\cal H}_{A|A}=\emptyset$.

\begin{figure}[h]
	\centering
	{\includegraphics[scale=0.665]{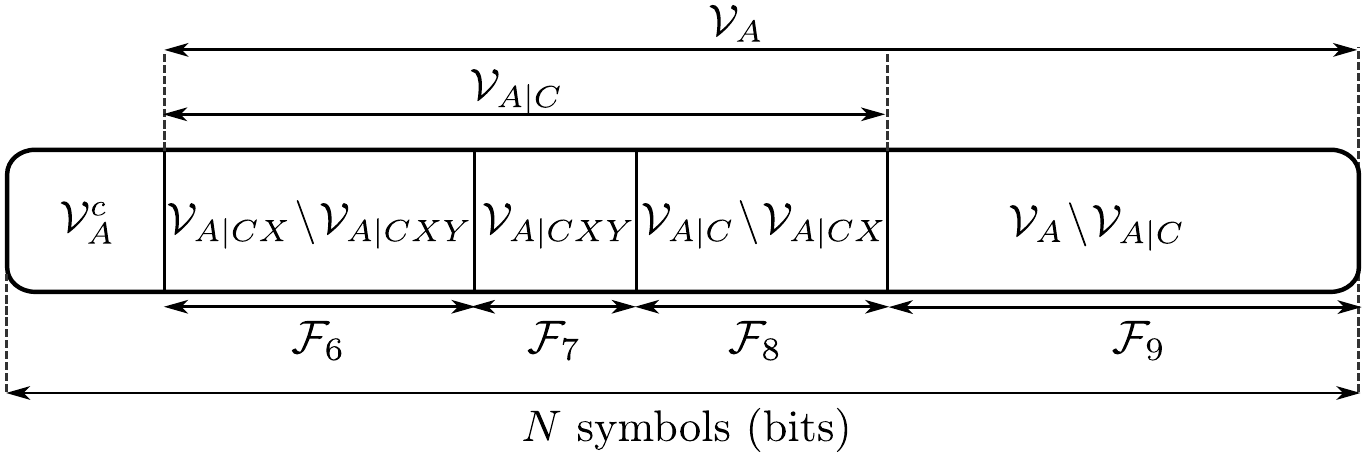}}
	\caption{Index sets for codeword $A$.}
	\label{fig:CoordiationSets2}
\end{figure}

Similarly, in terms of the polarization sets in \eqref{eq:A_Sets1}-\eqref{eq:A_Sets4} we define the sets combining channel resolvability for strong coordination and broadcast channel construction as shown in Fig.~\ref{fig:CoordiationSets2}
\begin{align*}
{\cal F}_6&\triangleq {\cal V}_{A|CX}\setminus {\cal V}_{A|CXY},& {\cal F}_8&\triangleq {\cal V}_{A|C}\setminus {\cal V}_{A|CX},\\
{\cal F}_7&\triangleq {\cal V}_{A|CXY},& {\cal F}_9&\triangleq {\cal V}_{A}\setminus {\cal V}_{A|C}.\\ 
\end{align*}

Finally, we define the sequence $T^N$ as the polar linear transformation of $Y^N$ i.e., $T^{N}\triangleq Y^{N}{\bf G}_n$. Now consider $Y^{N}= T^{N} {\bf G}_n$. By invertibility of ${\bf G}_n$ we define the very high entropy set:
\begin{equation}\label{eq:Y_Set}
{\cal V}_{Y|BC} \triangleq \{i\in \llbracket 1,N \rrbracket: H(T_i|T^{i-1}B^NC^N)\!>\! \log|\mathcal{Y}|-\delta_N\},
\end{equation}
satisfying
\begin{equation*} 
\lim_{N\rightarrow \infty} \frac{|{\cal V}_{Y|BC}|}{N} = H(Y|BC). 
\end{equation*}

This set is useful for expressing the randomized generation of $Y^N$ via simulating the channel $P_{Y|BC}$ in Fig.~\ref{fig:StrongCoordination} as a source polarization operation \cite{SPolarztion2010Arikan,chou2016empirical}. Note that here, we let ${\bf G}_n$ be a polar code generator matrix defined appropriately based on the alphabet of $Y$, e.g., if $|\cal Y|$ is the prime number $q\geq2$, ${\bf G}_n$ is as defined in
	Section~\ref{sec:notation}. However, the matrix operation is now
        carried out in the Galois field $GF(q)$ and the entropy terms of the polarization sets are calculated with respect to base-$q$ logarithms \cite[Theorem~4]{SPolarztion2010Arikan}. We now proceed to describe the encoding and decoding algorithms.

\subsubsection{Encoding}
The encoding protocol described in Algorithm~1 is performed over $k
\in \mathbb{N}$ blocks of length $N$ resulting in a storage complexity
  of ${\cal O}(kN)$ and a time complexity of ${\cal O}(kN\log N)$. In
  Algorithm~1 we use the \emph{tilde} notation (i.e.,~$\widetilde{U}_1^N,
  \widetilde{U}_2^N,$ $\widetilde{A}^N,$ and $\widetilde{C}^N$) to denote
  the change in the statistics of the length-$N$ random variables
  (i.e.,~$U_1^N, U_2^N,$${A}^N,$ and ${C}^N$) as a result of inserting
  uniformly distributed message and randomness bits at specific indices
  during encoding. Since for strong coordination
the goal is to approximate a target joint distribution with the minimum amount of randomness, the encoding scheme performs channel resolvability while reusing a fraction of the common randomness over several blocks (i.e.,  
randomness recycling) as in \cite{chou2016empirical}. The encoding scheme also leverages a
block chaining construction
\cite{PC_BCMondelli2014,PC_BCRameBloch16,PCempirical2016,GKLJB16} to
achieve the rates stated in Theorem~1.

\begin{figure}[h]
	\hspace*{-10ex}\centering
	{\includegraphics[scale=0.6]{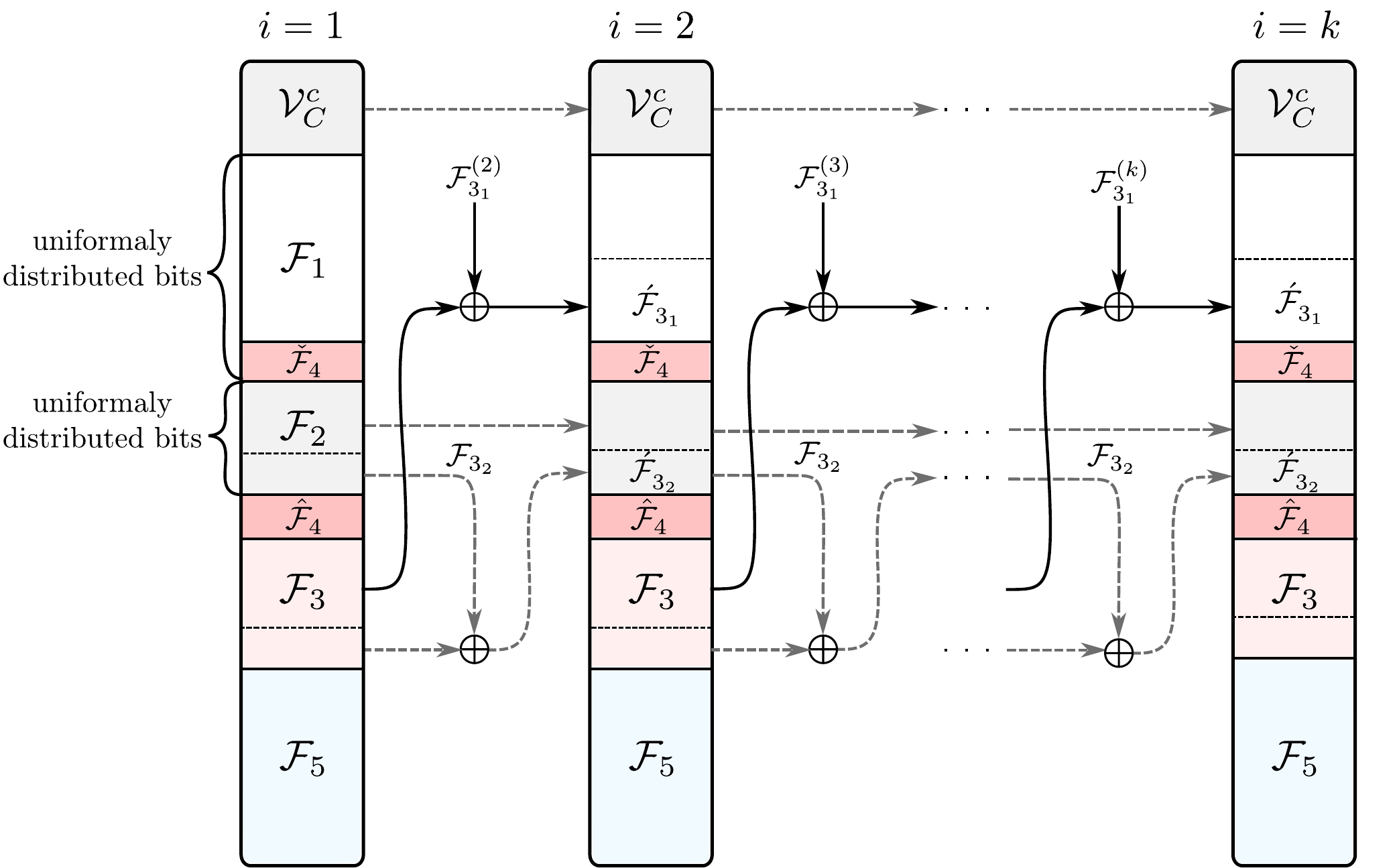}}
	\caption{Chaining construction for block encoding.}
	\label{fig:ChanningConstruction}
\end{figure}

\begin{table}[h!]
	\normalsize
	\begin{tabular}{p{0.95\linewidth}}
		\specialrule{.1em}{.05em}{.05em} 
		\rule{0pt}{2.5ex} 
		\hspace{-0.5em}\noindent\textbf{Algorithm 1:} Encoding algorithm at Node $\mathsf X$ for strong coordination\\
		\specialrule{.1em}{.05em}{.05em} \specialrule{.1em}{.05em}{.05em} 
		\rule{0pt}{2.5ex}  
		\hspace{-0.35em}\textbf{Input:} $X_{1:k}^N$, uniformly distributed local randomness bits $M_{1_{1:k}}$ of size $k|{\cal F}_6|$, common randomness bits reused over $k$ blocks $\bar{J}=(\bar{J}_1,\bar{J}_2)$ of sizes $|{\cal F}_2\cup \hat{\cal F}_4|$, and $|{\cal F}_7|$, respectively, and uniformly distributed common randomness bits for each block $J_{1:k}$, each of size $k|\check{\cal F}_4\cup {\cal F}_1|$, shared with Node $\mathsf Y$.
		
		\textbf{Output:} $\widetilde{A}_{1:k}^N$\\
		1.~\textbf{for} $i=2,\dots,k$ \textbf{do}\\
		2.~${\cal E}_2$ in Fig.~\ref{fig:BroadcastChannelPC} constructs $\widetilde{U}_{2_i}^N$ bit-by-bit as follows:\\
		~~\textbf{if} $i=1$ \textbf{then} 
		\begin{itemize}
			\item $\widetilde{U}_{2_i}^N[{\cal F}_1 \cup \check{\cal F}_4] \leftarrow {J}_{i}$ 
			\item $\widetilde{U}_{2_i}^N[{\cal F}_2 \cup \hat{\cal F}_4 ] \leftarrow \bar{J}_1$
		\end{itemize}
		~~\textbf{else}
		\begin{itemize}
			\item Let ${\cal F}_{3_1}^{(i)}$, ${\cal F}_{3_2}$ be sets of the size $(|{\cal F}_m|\times|{\cal F}_3|)/(|{\cal F}_1|+ |{\cal F}_2|)$ for $m\in\{ 1,2\}.$
			\item $\big(\widetilde{U}_{2_i}^N[({\cal F}_1\setminus\acute{\cal F}_{3_1}) \cup \check{\cal F}_4], {\cal F}_{3_1}^{(i)}\big) \leftarrow {J}_{i}$ 
			\item $\big(\widetilde{U}_{2_i}^N[({\cal F}_2\setminus\acute{\cal F}_{3_2}) \cup \hat{\cal F}_4 ],{\cal F}_{3_2} \big) \leftarrow \bar{J}_1$
			\item $\widetilde{U}_{2_i}^N[\acute{\cal F}_{3_1} ] \leftarrow \widetilde{U}_{2_{i-1}}^N[{\cal F}_3\setminus {\cal F}_{3_2}] \oplus {\cal F}_{3_1}^{(i)} $
			\item $\widetilde{U}_{2_i}^N[\acute{\cal F}_{3_2}] \leftarrow \widetilde{U}_{2_{i-1}}^N[{\cal F}_3\setminus {\cal F}_{3_1}] \oplus {\cal F}_{3_2} $
		\end{itemize}
		~~\textbf{end}
		\begin{itemize}
			\item Given $X_i^N$, successively draw the remaining components of $\widetilde{U}_{2_i}^N$  according to $\tilde{P}_{U_{2_i,j}|U_{2_i}^{j-1}X_i^N}$ defined by
			\begin{equation}\label{eq:MsgEncC}
			\hspace*{-1.7ex} \begin{aligned}
			\tilde{P}_{U_{2_i,j}|U_{2_i}^{j-1}X_i^N} \triangleq  \begin{cases}
			{Q}_{U_{2,j}|U_{2}^{j-1}} & \!j \in  {\cal V}_{C}^{c}, \\
			{Q}_{U_{2,j}|U_{2}^{j-1}X^N} & \!j \in {\cal F}_3 \cup {\cal F}_5.
			\end{cases}
			\end{aligned}
			\end{equation} 
		\end{itemize}
		3.~$\widetilde{C}_{i}^{N} \leftarrow \widetilde{U}_{2_i}^{N} {\bf G}_n $ \\
		4.~${\cal E}_1$ in Fig.~\ref{fig:BroadcastChannelPC} constructs $\widetilde{U}_{1_i}^N$ bit-by-bit as follows: 
		\begin{itemize}
			\item $\widetilde{U}_{1_i}^N[{\cal F}_6] \leftarrow {M}_{1_i}$
			\item $\widetilde{U}_{1_i}^N[{\cal F}_7] \leftarrow \bar{J}_2$
			\item Given $X_i^N$ and $ \widetilde{C}_{i}^{N}$, successively draw the remaining components of $\widetilde{U}_{1_i}^N$ according to $\tilde{P}_{U_{1_i,j}|U_{1_i}^{j-1}C_i^NX_i^N}$ defined by
			\vspace{-1.5ex}\begin{equation}\label{eq:MsgEncA}
			\hspace*{-2.52ex} \begin{aligned}
			\tilde{P}_{U_{1_i,j}|U_{1_i}^{j-1}C_i^NX_i^N} \triangleq  \begin{cases}
			{Q}_{U_{1,j}|U_{1}^{j-1}} & j \in  {\cal V}_{A}^{c}, \\
			{Q}_{U_{1,j}|U_{1}^{j-1}C^N} & j \in {\cal F}_9, \\
			{Q}_{U_{1,j}|U_{1}^{j-1}C^NX^N} & j \in {\cal F}_8.
			\end{cases}
			\end{aligned}
			\end{equation}
		\end{itemize}
		5.~$\widetilde{A}_i^{N} \leftarrow \widetilde{U}_{1_i}^{N} {\bf G}_n $\\
		6.~Transmit $\widetilde{A}_{i}^N$\\
		7.~\textbf{end for}\\
		\specialrule{.1em}{.05em}{.05em}
	\end{tabular}
\end{table}
\vspace{-0.5ex}

More precisely, as demonstrated in Fig.~\ref{fig:StrongCoordination}, we are interested
in successfully recovering the message ${I}$ that is intended for the channel of the weak user $P_{B|A}$ in Fig.~\ref{fig:BroadcastChannelPC}. However, the challenge is to communicate the set ${\cal F}_3$ that includes bits of the message $I$ that are corrupted by the channel
noise. This suggests that we apply a variation of block chaining
only at  encoder ${\cal E}_2$, generating the codeword $C^N$ as
follows (see  Fig.~\ref{fig:ChanningConstruction}).
At encoder ${\cal E}_2$, the set ${\cal F}_3$
of block $i\in \llbracket 1,k \rrbracket$ is embedded in the reliably decodable bits of ${\cal F}_1
\cup {\cal F}_2$ of the following block $i+1$. This is possible by
following the  decodability constraint (see \eqref{equ:JSRR4},
\eqref{equ:JSRR5} of Theorem 1) that ensures that the size of the set
${\cal F}_3$ is smaller than the combined size of the sets ${\cal F}_1$
and ${\cal F}_2$ \cite{PCempirical2016}. However, since
these sets originally contain uniformly distributed
common randomness  $J$ \cite{chou2016empirical}, the bits of ${\cal F}_3$ can be
embedded while maintaining the uniformity of the randomness  by taking advantage of the Crypto Lemma \cite[Lemma
2]{Forney2004role}, \cite[Lemma 3.1]{Bloch2011}. Then, to ensure that ${\cal F}_3$ is equally distributed over ${\cal F}_1 \cup {\cal F}_2$, ${\cal
	F}_3$ is partitioned according to the ratio between $|{\cal F}_1|$ and $|{\cal
	F}_2|$. To utilize the Crypto Lemma, we introduce $\mathcal{F}_{3_2}$ and
$\mathcal{F}_{3_1}^{(i)}$, which represent uniformly distributed common 
randomness used to randomize the information bits of
$\mathcal{F}_3$. The difference is that $\mathcal{F}_{3_2},$ as ${\cal F}_2,$ represents a fraction of common randomness that can
be reused over $k$ blocks, whereas a realization of the randomness
in $\mathcal{F}_{3_{1}}^{(i)}$ needs to be
provided in each new block. Note that, as visualized in Fig.~6, both the
subsets $\acute{\cal F}_{3_1}\subset {\cal F}_1$ and $\acute{\cal
	F}_{3_2} \subset {\cal F}_2$ represent the resulting uniformly distributed bits of ${\cal F}_{3}$ of the previous block, where $|\acute{\cal F}_{3_1}|=|{\cal F}_{3_1}^{}|$ and $|\acute{\cal F}_{3_2}|=|{\cal F}_{3_2}|.$
Finally, in an additional block $k+1$ we use a good channel  code to
reliably transmit the set ${\cal F}_3$ of the last block $k.$ Note that
  since uniformly random bits are reused  to convey information bits,
  chaining can be seen as a derandomization strategy.

\subsubsection{Decoding}
The decoder is described in Algorithm 2. In Algorithm~2, we use the
  \emph{hat} notation, i.e.,~ $\widehat{U}_2^N$ and $\widehat{C}^N$, to
  distinguish the reconstruction of the $N$-length random variables,
  i.e.,~$U_2^N$ and consequently $C_2^N$, from the corresponding quantities at the encoder.
Recall that we are only interested
in the message $\hat{I}$ intended for the weak user channel given by
$P_{B|A}$ in Fig.~\ref{fig:BroadcastChannelPC}. As a result, we
only state the decoding protocol at ${\cal D}_2$ that recovers the codeword $\widehat{C}^{N}.$ Note that the
decoding is done in reverse order after receiving the extra $k+1$
block containing the bits of set ${\cal F}_3$ of the last block
$k$. In particular, in each block $i\in[1,k-1]$ the bits in
$\mathcal{F}_3$ are obtained by successfully recovering the bits in both
$\mathcal{F}_1$ and $\mathcal{F}_2$ in block $i+1$.

\vspace{0.5ex}
\begin{table}[h!]
	\normalsize
	\begin{tabular}{p{0.95\linewidth}}
          \specialrule{.1em}{.05em}{.05em} 
          \rule{0pt}{2.5ex} 
          \hspace{-0.5em}\noindent\textbf{Algorithm 2:} Decoding algorithm at Node $\mathsf Y$ for strong coordination\\
          \specialrule{.1em}{.05em}{.05em}	\specialrule{.1em}{.05em}{.05em}
          \rule{0pt}{2.5ex}  
          \hspace{-0.35em}\textbf{Input:} $B_{1:k}^N,$ uniformly
          distributed common randomness $\bar{J}_1$ of sizes $|{\cal
          F}_2\cup \hat{\cal F}_4|$ reused over $k$ blocks, ``fresh'' uniformly distributed common randomness $J_{1:k}$ each of size $k|\check{\cal F}_4\cup {\cal F}_1|$  for all $k$ blocks and shared with Node $\mathsf X$.\\ 
		
		\textbf{Output:} $\widetilde{Y}_{1:k}^N$\\
		1.~\textbf{For} block $i=k,\dots,1$ \textbf{do}\\ 
		2.~${\cal D}_2$ in Fig.~\ref{fig:BroadcastChannelPC} constructs $\widehat{U}_{2_i}^N$ bit-by-bit as follows:
		\begin{itemize}
			\item $\big(\widehat{U}_{2_i}^N[({\cal F}_1\setminus\acute{\cal F}_{3_1}^{}) \cup \check{\cal F}_4], {\cal F}_{3_1}^{(i)}\big) \leftarrow {J}_{i}$ 
			\item $\big(\widehat{U}_{2_i}^N[({\cal F}_2\setminus\acute{\cal F}_{3_2}) \cup \hat{\cal F}_4 ],{\cal F}_{3_2}\big) \leftarrow \bar{J}_1$
			\item Given $B_i^N$ successively draw the components of  $\widehat{U}_{2_i}^N$ according to 
			$\tilde{P}_{U_{2_i,j}|U_{2_i}^{j-1},B_i^N}$ defined by
			\begin{equation}\label{eq:MsgDec}
			\hspace{-5ex}\begin{aligned}
			\tilde{P}_{U_{2_i,j}|U_{2_i}^{j-1}B_i^N}\triangleq\begin{cases}
			{Q}_{U_{2,j}|U_{2}^{j-1}}\! & \!j\! \in  {\cal V}_{C}^{c}, \\
			{Q}_{U_{2,j}|U_{2}^{j-1}B_i^N}\!\! & \!j\! \in \acute{\cal F}_{3_2}\!\cup\!\acute{\cal F}_{3_1}\!\cup\!{\cal F}_5.
			\end{cases}
			\end{aligned}\hspace{-1ex}
			\end{equation}
		\end{itemize}
		3.~\textbf{if} $i=k$ \textbf{then}
		\begin{itemize}
			\item $\widehat{U}_{2_i}^N[{\cal F}_3]\leftarrow B_{k+1}^N $
		\end{itemize}
		~~\textbf{else}
		\begin{itemize}
			\item $ \widehat{U}_{2_{i}}^N[{\cal F}_3\setminus {\cal F}_{3_2}]  \leftarrow \widehat{U}_{2_{i+1}}^N[\acute{\cal F}_{3_1} ]\oplus {\cal F}_{3_1}^{(i+1)} $
			\item $\widehat{U}_{2_{i}}^N[{\cal F}_3\setminus {\cal F}_{3_1}]   \leftarrow  \widehat{U}_{2_{i+1}}^N[\acute{\cal F}_{3_2}]\oplus {\cal F}_{3_2} $
		\end{itemize}
		4.~Let \begin{itemize}
			\item $\widehat{U}_{2_i}^N[\acute{\cal F}_{3_1}^{} ] \leftarrow {\cal F}_{3_1}^{(i)}$ 
			\item $\widehat{U}_{2_i}^N[\acute{\cal F}_{3_2}] \leftarrow {\cal F}_{3_2} $
		\end{itemize}	
		5.~$\widehat{C}_{i}^{N} \leftarrow \widehat{U}_{2_i}^{N} {\bf G}_n $ \\
		6.~Channel simulation: given $\widehat{C}_{i}^N$ and ${B_i}^N,$ successively draw the components of $\widetilde{T}_{i}^{N}$ according to 
		\begin{equation}\label{eq:seqGen}
		\hspace{-1.16ex}\begin{aligned}
		\tilde{P}_{T_{i,j}|T_{i}^{j-1}B_i^NC_i^N} \triangleq  \begin{cases}
		1/|{\cal Y}| & \!j \in {\cal V}_{Y|BC}, \\
		{Q}_{T_{j}|T_{}^{j-1}B^NC^N} & \!j \in {\cal V}_{Y|BC}^{c}.
		\end{cases}
		\end{aligned}
		\end{equation}
		5.~$\widetilde{Y}_i^N \leftarrow \widetilde{T}_i^N {\bf G}_n$\\
		6.~\textbf{end for}\\
		\specialrule{.1em}{.05em}{.05em} 
	\end{tabular}
\end{table}
\vspace{-0.5ex}

\subsection{Scheme Analysis}
We now provide an analysis of the coding scheme of Section~\ref{sec:PolarCodingScheme}. Let the joint pmf of actions induced by the polar coding scheme be $\tilde{P}_{X^NY^N}$. For strong coordination coding scheme $\tilde{P}_{X^NY^N}$ must be \emph{close} in total variation to the $N$ i.i.d. copies of desired joint pmf $(X,Y)\sim Q_{XY},$ 
$Q^N_{XY}$, i.e.,
\begin{equation}\label{eq:PolarStrngCoorCondtion}
\norm{\tilde{P}_{X^NY^N}-Q^N_{XY}}_{{\scriptscriptstyle TV}} 
<\epsilon.
\end{equation}
The analysis is based on the KL divergence which upper bounds the total variation in \eqref{eq:PolarStrngCoorCondtion} by Pinsker's inequality. We start the analysis with a set of sequential lemmas. In particular, Lemma~\ref{lemma1} is useful to show in Lemma~\ref{lemma2} that the strong coordination scheme based on channel resolvability holds for each block individually regardless of the randomness recycling. Note that, in the current Section~\ref{sec:PolarCode} and the associated Appendices~\ref{Appnx:ProofLemma3} and \ref{Appnx:ProofLemma4}, we refrain from using the $N$-fold product notation of joint, and conditional distribution, e.g., respectively, $Q_{X^NY^N}=Q^N_{XY}$ and ${Q}_{U_{2}^N|X^N}={Q}^N_{U_{2}|X}$ to unify the notation across conditional distributions.

\begin{lemma}\label{lemma1}
	For block $i \in \llbracket 1,k \rrbracket ,$ we have
	\begin{equation*}
	\mathbb{D}(Q_{A^NC^NX^N}||\tilde{P}_{A_i^NC_i^NX_i^N}) \leq 2N\delta_{N}.
	\end{equation*} 
\end{lemma}

\begin{proof}
	We have
	\begin{equation*}
	\begin{split} 
	&\mathbb{D}({Q}_{A^NC^NX^N}||\tilde{P}_{A_i^NC_i^NX_i^N}) 
	\stackrel{(a)}{=} \mathbb{D}({Q}_{U_{1}^NU_{2}^NX^N}||\tilde{P}_{U_{1_i}^NU_{2_i}^NX_i^N}) \\
	&\stackrel{}{=} \mathbb E_{Q_{X^N}}\Big[ \mathbb{D}({Q}_{U_{1}^NU_{2}^N|X^N}||\tilde{P}_{U_{1_i}^NU_{2_i}^N|X_i^N}) \Big] \\
	&\stackrel{}{=} \mathbb E_{Q_{X^N}}\Big[ \mathbb{D} ({Q}_{U_{2}^N|X^N}{Q}_{U_{1}^N|U_{2}^NX^N}||\tilde{P}_{U_{2_i}^N|X_i^N}\tilde{P}_{U_{1_i}^N|U_{2_i}^NX_i^N}) \Big] \\
	&\stackrel{(b)}{=} \mathbb E_{Q_{X^N}}\Big[ \mathbb{D}({Q}_{U_{2}^N|X^N}||\tilde{P}_{U_{2_i}^N|X_i^N}) +\mathbb{D}({Q}_{U_{1}^N|U_{2}^NX^N}||\tilde{P}_{U_{1_i}^N|U_{2_i}^NX_i^N}) \Big] \\
	&\stackrel{(c)}{=} \sum_{j=1}^{N} \mathbb E_{Q_{U_{2}^{j-1}X^N}}\Big[ \mathbb{D} ({Q}_{U_{2,j}|U_{2}^{j-1}X^N}||\tilde{P}_{U_{2_i,j}|U_{2_i}^{j-1}X_i^N}) \Big]+\! \sum_{j=1}^{N} \mathbb E_{Q_{U_{1}^{j-1}\!U_{2}^{N}\!X^N}}\!\Big[ \mathbb{D} ({Q}_{U_{1,j}|U_{1}^{j-1}\!U_{2}^N\!X^N}\!||\tilde{P}_{U_{1_i,j}|U_{1_i}^{j-1}\!U_{2_i}^N\!X_i^N}) \Big]\\ 
	&\stackrel{(d)}{=}\sum_{j \notin {\cal F}_3 \cup {\cal
			F}_5}\!\! \mathbb E_{Q_{U_{2}^{j-1}X^N}}\Big[ \mathbb{D}
	({Q}_{U_{2,j}|U_{2}^{j-1}X^N}||\tilde{P}_{U_{2_i,j}|U_{2_i}^{j-1}X_i^N})
	\Big] +\sum_{j \notin {\cal F}_8}\mathbb E_{Q_{U_{1}^{j-1}\!U_{2}^{N}\!X^N}}\!\Big[ \mathbb{D} ({Q}_{U_{1,j}|U_{1}^{j-1}\!U_{2}^N\!X^N}||\tilde{P}_{U_{1_i,j}|U_{1_i}^{j-1}\!U_{2_i}^N\!X_i^N}) \Big]\\
	&\stackrel{(e)}{=} \sum_{j \in {\cal V}_{C}^{c}\cup{\cal V}_{C|X}} \!\!\mathbb E_{Q_{U_{2}^{j-1}X^N}} \Big[ \mathbb{D} ({Q}_{U_{2,j}|U_{2}^{j-1}X^N}||\tilde{P}_{U_{2_i,j}|U_{2_i}^{j-1}X_i^N}) \Big]\\
	&\hspace*{18ex} +\sum_{j \in {\cal V}_{A}^{c}\cup {\cal V}_{A|CX} \cup {\cal V}_{A}\!\setminus\!{\cal V}_{A|C} } \mathbb E_{Q_{U_{1}^{j-1}U_{2}^{N}X^N}} \Big[ \mathbb{D} ({Q}_{U_{1,j}|U_{1}^{j-1}U_{2}^NX^N}||\tilde{P}_{U_{1_i,j}|U_{1_i}^{j-1}U_{2_i}^NX_i^N}) \Big]\\
		&\stackrel{(f)}{=} \sum_{j \in {\cal V}_{C}^{c}} \Big( H(U_{2,j}|U_{2}^{j-1}) - H(U_{2,j}|U_{2}^{j-1}X^N) \Big)+\sum_{j \in {\cal V}_{C|X}} \Big( 1-H(U_{2,j}|U_{2}^{j-1}X^N) \Big)\\ 
	& \hspace*{18ex} +\sum_{j \in {\cal V}_{A}^{c}} \Big( H(U_{1,j}|U_{1}^{j-1}) -H(U_{1,j}|U_{1}^{j-1}U_2^NX^N) \Big)+\sum_{j \in {\cal V}_{A|CX}} \!\Big( 1- H(U_{1,j}|U_{1}^{j-1}U_2^NX^N) \Big)\\
	& \hspace*{18ex} + \sum_{j \in {\cal V}_{A|C}^{c} \!\setminus\!{\cal V}_{A}^{c}} \Big( H(U_{1,j}|U_{1}^{j-1}U_{2}^N) - H(U_{1,j}|U_{1}^{j-1}U_{2}^NX^N) \Big)\\
	&\stackrel{(g)}{=} \sum_{j \in {\cal V}_{C}^{c}} \Big( H(U_{2,j}|U_{2}^{j-1}) - H(U_{2,j}|U_{2}^{j-1}X^N) \Big)+\sum_{j \in {\cal V}_{C|X}} \Big( 1-H(U_{2,j}|U_{2}^{j-1}X^N) \Big)\\
	&\hspace*{18ex}+\sum_{j \in {\cal V}_{A}^{c}} \Big( H(U_{1,j}|U_{1}^{j-1}) -H(U_{1,j}|U_{1}^{j-1}C^NX^N) \Big)+\sum_{j \in {\cal V}_{A|CX}} \!\Big( 1- H(U_{1,j}|U_{1}^{j-1}C^NX^N) \Big)\\
	&\hspace*{18ex}+\sum_{j \in {\cal V}_{A|C}^{c} \!\setminus\!{\cal V}_{A}^{c}} \Big( H(U_{1,j}|U_{1}^{j-1}C^N) - H(U_{1,j}|U_{1}^{j-1}C^NX^N) \Big)\\
	&\stackrel{(h)}{\leq} (|{\cal V}_{C}^{c}|+|{\cal V}_{C|X}|+ |{\cal V}_{A|XC}|+|{\cal V}_{A|C}^{c}|)\delta_N \leq 2N\delta_N,\\
	\end{split}
	\end{equation*}
	
	\noindent where
	\begin{itemize} 
		\item[($a$)] holds by invertibility of ${\bf G}_n$;
		\item[($b$)]\hspace*{-0.5em}\;-\;($c$)~follows from the chain rule of the KL divergence \cite{EoIT:2006};
		\item[($d$)] results from the definitions of the conditional distributions in \eqref{eq:MsgEncC}, and \eqref{eq:MsgEncA};
		\item[($e$)] follows from the definitions of the index sets as shown in Figs.~\ref{fig:CoordiationSets1} and~\ref{fig:CoordiationSets2};
		\item[($f$)] results from the encoding of $\widetilde{U}_{1_i}^N$ and $\widetilde{U}_{2_i}^N$ bit-by-bit at ${\cal E}_1$ and ${\cal E}_2,$ respectively, with uniformly distributed randomness bits and message bits. These bits are generated by applying successive cancellation encoding using previous bits and side information with conditional distributions defined in \eqref{eq:MsgEncC} and \eqref{eq:MsgEncA};
		\item[($g$)] holds by the one-to-one relation between ${U}_{2}^N$ and $C^N$;
		\item[($h$)] follows from the sets defined in~\eqref{eq:C_Sets} and~\eqref{eq:A_Sets}.
	\end{itemize}
\end{proof}

\begin{lemma}\label{lemma2}
	For block $i \in \llbracket 1,k \rrbracket ,$ we have
	\begin{align*}
	\mathbb{D}(&\tilde{P}_{X_i^NY_i^N}||Q_{X^NY^N})\leq\mathbb{D}(\tilde{P}_{X_i^NA_i^NC_i^NB_i^N\widehat{C}_i^NY_i^N}||Q_{X^NA^NC^NB^N\widehat{C}^NY^N})\leq \delta_{N}^{(1)} 
	\end{align*}
	where $\delta_{N}^{(1)}\triangleq {\cal O}(\sqrt{N^3 \delta_N}).$ 
\end{lemma}

\begin{proof}
	Consider the following argument. 
		\begin{equation}
		\begin{split} \notag
		\mathbb{D}&(\tilde{P}_{X_i^NA_i^NC_i^NB_i^N\widehat{C}_i^NY_i^N}||Q_{X^NA^NC^NB^N\widehat{C}^NY^N})\\ 
		&\qquad\stackrel{}{=} \mathbb{D}(\tilde{P}_{Y_i^N|X_i^NA_i^NC_i^NB_i^N\widehat{C}_i^N}\tilde{P}_{X_i^NA_i^NC_i^NB_i^N\widehat{C}_i^N}||Q_{Y^N|X^NA^NC^NB^N\widehat{C}^N}Q_{X^NA^NC^NB^N\widehat{C}^N}) \\
		&\qquad\stackrel{(a)}{=} \mathbb{D}(\tilde{P}_{Y_i^N|B_i^N\widehat{C}_i^N}\tilde{P}_{X_i^NA_i^NC_i^NB_i^N\widehat{C}_i^N}||Q_{Y^N|B^N\widehat{C}^N}Q_{X^NA^NC^NB^N\widehat{C}^N}) \\
		&\qquad\stackrel{}{=} \mathbb{D}(\tilde{P}_{Y_i^N|B_i^N\widehat{C}_i^N}\tilde{P}_{B_i^N\widehat{C}_i^N|X_i^NA_i^NC_i^N}\tilde{P}_{X_i^NA_i^NC_i^N}||Q_{Y^N|B^N\widehat{C}^N}Q_{B^N\widehat{C}^N|X^NA^NC^N}Q_{X^NA^NC^N}) \\
		&\qquad\stackrel{(b)}{=} \mathbb{D}(\tilde{P}_{Y_i^N|B_i^N\widehat{C}_i^N}\tilde{P}_{B_i^N\widehat{C}_i^N|A_i^NC_i^N} \tilde{P}_{X_i^NA_i^NC_i^N}||Q_{Y^N|B^N\widehat{C}^N}Q_{B^N\widehat{C}^N|A^NC^N}Q_{X^NA^NC^N}) \\
		&\qquad\stackrel{(c)}{\leq} {\delta}_N^{(2)} + \mathbb{D}(\tilde{P}_{Y_i^N|B_i^N\widehat{C}_i^N}\tilde{P}_{B_i^N\widehat{C}_i^N|A_i^NC_i^N}\tilde{P}_{X_i^NA_i^NC_i^N}||\tilde{P}_{Y_i^N|B_i^N\widehat{C}_i^N}\tilde{P}_{B_i^N\widehat{C}_i^N|A_i^NC_i^N}Q_{X^NA^NC^N}) \\
		&\qquad\hspace*{18ex} + \mathbb{D}(\tilde{P}_{Y_i^N|B_i^N\widehat{C}_i^N}\tilde{P}_{B_i^N\widehat{C}_i^N|A_i^NC_i^N}Q_{X^NA^NC^N}||Q_{Y^N|B^N\widehat{C}^N}Q_{B^N\widehat{C}^N|A^NC^N}Q_{X^NA^NC^N})\\
		&\qquad\stackrel{(d)}{=} {\delta}_N^{(2)} + \mathbb{D}(\tilde{P}_{X_i^NA_i^NC_i^N}||Q_{X^NA^NC^N}) 
		+\mathbb{D}(\tilde{P}_{Y_i^N|B_i^N\widehat{C}_i^N}\tilde{P}_{B_i^N\widehat{C}_i^N|A_i^NC_i^N}||Q_{Y^N|B^N\widehat{C}^N}Q_{B^N\widehat{C}^N|A^NC^N})\\
		&\qquad\stackrel{(e)}{\leq} {\delta}_N^{(2)} + \hat{\delta}_N^{(2)}
		+\mathbb{D}(\tilde{P}_{Y_i^N|B_i^N\widehat{C}_i^N}\tilde{P}_{B_i^N\widehat{C}_i^N|A_i^NC_i^N}||Q_{Y^N|B^N\widehat{C}^N}Q_{B^N\widehat{C}^N|A^NC^N})\\
		&\qquad\stackrel{(f)}{=} {\delta}_N^{(2)} + \hat{\delta}_N^{(2)}
		+\mathbb{D}(\tilde{P}_{Y_i^N|B_i^N\widehat{C}_i^N}||Q_{Y^N|B^N\widehat{C}^N})+\mathbb{D}(\tilde{P}_{B_i^N\widehat{C}_i^N|A_i^NC_i^N}||Q_{B^N\widehat{C}^N|A^NC^N})\\
		&\qquad\stackrel{(g)}{\leq} {\delta}_N^{(2)} +  \hat{\delta}_N^{(2)}
		-N\log(\mu_{YB\widehat{C}})\sqrt{2\ln 2} \sqrt{\mathbb{D}(Q_{Y^N|B^N\widehat{C}^N}||\tilde{P}_{Y_i^N|B_i^N\widehat{C}_i^N})}\\
		&\qquad\hspace*{18ex}-N\log(\mu_{ACB\widehat{C}})\sqrt{2\ln 2} \sqrt{\mathbb{D}(Q_{B^N\widehat{C}^N|A^NC^N}||\tilde{P}_{B_i^N\widehat{C}_i^N|A_i^NC_i^N})} \\
		&\qquad\stackrel{(h)}{\leq} {\delta}_N^{(2)} + \hat{\delta}_N^{(2)} -N\log(\mu_{YB\widehat{C}})\sqrt{2\ln 2} \sqrt{N \delta_N} -N\log(\mu_{ACB\widehat{C}})\sqrt{2\ln 2} \sqrt{N \delta_N} \\
		\end{split}
		\end{equation}
	
	\vspace{1.5ex}
	\noindent In this argument: 
	\begin{itemize} 
		\item[($a$)]\hspace*{-0.5em}\;-\;($b$) results from the Markov chain $X^N\!-\!A^NC^N\!-\!B^N\widehat{C}^N\!-\!Y^N$; 
		\item[($c$)] follows from \cite[Lemma 16]{chou2016empirical} where 
		\begin{align*} {\delta}_N^{(2)} &\triangleq -N\log(\mu_{XACB\widehat{C}Y})\sqrt{2\ln 2} \sqrt{2N\delta_{N}},\\
		\mu_{XACB\widehat{C}Y}&\triangleq {\textstyle\min^{*}_{x,y,a,c,b,\hat{c}}} \big(Q_{XACB\widehat{C}Y}\big);
		\end{align*}
		\item[($d$)] follows from the chain rule of KL divergence \cite{EoIT:2006};
		\item[($e$)] holds by Lemma~\ref{lemma1} and \cite[Lemma 14]{chou2016empirical} where 
		\begin{align*} \hat{\delta}_N^{(2)} &\triangleq -N\log(\mu_{XAC})\sqrt{2\ln 2} \sqrt{2N\delta_{N}},\\
		\mu_{XAC} &\triangleq {\textstyle\min^{*}_{x,a,c}} \big(Q_{XAC}\big);\\
		\end{align*}
		
		\item[($f$)] follows from the chain rule of KL divergence \cite{EoIT:2006};
		\item[($g$)] holds by \cite[Lemma 14]{chou2016empirical}, where
		\begin{align*}\mu_{ACB\widehat{C}} &\triangleq{\textstyle\min^{*}_{a,c,b,\hat{c}}} \big(Q_{ACB\widehat{C}}\big),\\
		\mu_{YB\widehat{C}} &\triangleq{\textstyle\min^{*}_{y,b,\hat{c}}} \big(Q_{YB\widehat{C}}\big);\end{align*}
		\item[($h$)] holds by bounding the terms $\mathbb{D}(Q_{B^N\widehat{C}^N|A^NC^N}||\tilde{P}_{B_i^N\widehat{C}_i^N|A_i^NC_i^N}),$ and $\mathbb{D}(Q_{Y^N|B^N\widehat{C}^N}||\tilde{P}_{Y_i^N|B_i^N\widehat{C}_i^N})$, as follows: 
		\begin{itemize}
			\item First, we show that $\mathbb{D}(Q_{B^N\widehat{C}^N|A^NC^N}||\tilde{P}_{B_i^N\widehat{C}_i^N|A_i^NC_i^N})\leq N\delta_N$ by the following argument:
		\begin{equation*}
		\begin{split} 
		\mathbb{D}&(Q_{B^N\widehat{C}^N|A^NC^N}||\tilde{P}_{B_i^N\widehat{C}_i^N|A_i^NC_i^N})\\
		&\stackrel{(a)}{=} \mathbb{D}(Q_{B^N|A^N}Q_{\widehat{C}^N|B^N}||{Q}_{B^N|A^N}\tilde{P}_{\widehat{C}_i^N|B_i^N})\\ 
		&\stackrel{}{=} \mathbb{D}(Q_{\widehat{C}^N|B^N}||\tilde{P}_{\widehat{C}_i^N|B_i^N})\\ 
		&\stackrel{(b)}{=} \mathbb{D}(Q_{\widehat{U}^N|B^N}||\tilde{P}_{\widehat{U}_i^N|B_i^N})\\
		&\stackrel{(c)}{=}\!\sum_{j=1}^{N} \mathbb E_{Q_{U_{2}^{j-1}B^N}}\Big[ \mathbb{D}(Q_{U_{2,j}|U_{2}^{j-1}B^N}||\tilde{P}_{U_{2_i,j}|U_{2_i}^{j-1}B_i^N}) \Big]\\
		&\stackrel{(d)}{=}\!\sum_{j \in {\cal V}_{C}^{c}}\!\mathbb E_{Q_{U_{2}^{j-1}B^N}}\Big[ \mathbb{D}(Q_{U_{2,j}|U_{2}^{j-1}B^N}||\tilde{P}_{U_{2_i,j}|U_{2_i}^{j-1}B_i^N}\!) \Big] +\sum_{\!\!j \in {\cal H}_{C|B}\cup {\cal V}_{C|X}\!\!}\!\!\!\!\!\!\!\!\mathbb E_{Q_{U_{2}^{j-1}\!B^N}}\!\Big[ \mathbb{D}(Q_{\!U_{2,j}|U_{2}^{j-1}\!B^N}\!||\tilde{P}_{\!U_{2_i,j}|U_{2_i}^{j-1}\!B_i^N}\!) \Big]\\
		&\stackrel{(e)}{=}\!\!\sum_{j \in {\cal V}_{C}^{c}}\!\!\!\Big( H(U_{2,j}|U_{2}^{j-1}) - H(U_{2,j}|U_{2}^{j-1}B^N) \Big)+\sum_{j \in {\cal H}_{C|B}\cup {\cal V}_{C|X}}\!\!\!\Big( 1 - H(U_{2,j}|U_{2}^{j-1}B^N) \Big)\\ 
		&\stackrel{(f)}{\leq} |{\cal V}_{C}^{c}|\delta_N +|{\cal H}_{C|B}\cup {\cal V}_{C|X}|\delta_N \leq N\delta_N,  
		\end{split}
		\end{equation*}
		\noindent where
		\begin{itemize} 
			\item[($a$)] results from the Markov chain $C-A-B-\widehat{C}$ and the fact that $\tilde{P}_{B_i^N|A_i^N}=Q_{B^N|A^N}$;
			\item[($b$)] holds by the one-to-one relation between ${U}_{2}^N$ and $C^N$;
			\item[($c$)] follows from the chain rule of KL divergence \cite{EoIT:2006};
			\item[($d$)]\hspace*{-0.5em}\,-\,($e$) results from the definitions of the conditional distributions in~\eqref{eq:MsgDec};
			\item[($f$)] follows from the sets defined in~\eqref{eq:C_Sets}. 
		\end{itemize}
	
	\item Next, we show that $\mathbb{D}(Q_{Y^N|B^N\widehat{C}^N}||\tilde{P}_{Y_i^N|B_i^N\widehat{C}_i^N}) \leq N\delta_N$ with the following derivation:
		\begin{equation*} 
		\begin{split} 
		\mathbb{D}(Q_{Y^N|B^N\widehat{C}^N}||\tilde{P}_{Y_i^N|B_i^N\widehat{C}_i^N}) 
		&\stackrel{(a)}{=}\sum_{j=1}^{N} \mathbb E_{Q_{T^{j-1}B_i^N\widehat{C}_i^N}}\Big[ \mathbb{D}(Q_{T_j|T^{j-1}B^N\widehat{C}^N}||\tilde{P}_{T_j|T^{j-1}B_i^N\widehat{C}_i^N}) \Big] \\
		&\stackrel{(b)}{=}\sum_{j \in {\cal V}_{Y|BC}}\mathbb E_{Q_{T^{j-1}B_i^N\widehat{C}_i^N}}\Big[ \mathbb{D}(Q_{T_j|T^{j-1}B^N\widehat{C}^N}||\tilde{P}_{T_j|T^{j-1}B_i^N\widehat{C}_i^N}) \Big] \\
		&\stackrel{(c)}{=}\sum_{j \in {\cal V}_{Y|BC}}\Big( \log|{\cal Y}| - H(T_j|T^{j-1}B^NC^N) \Big)\\ 
		&\stackrel{(d)}{\leq} |{\cal V}_{Y|BC}|\delta_N \leq N\delta_N, 
		\end{split}
		\end{equation*}
		\noindent where
		\begin{itemize} 
			\item[($a$)] follows from the chain rule of KL divergence \cite{EoIT:2006};
			\item[($b$)]\hspace*{-0.5em}\;-\;($c$) results from the definitions of the conditional distribution in~\eqref{eq:seqGen};
			\item[($d$)] follows from the set defined in~\eqref{eq:Y_Set}. 
		\end{itemize}
	\end{itemize}
	\end{itemize}
\end{proof}

Now, Lemmas~\ref{lemma3} and \ref{lemma4} provide the independence between two consecutive blocks and the independence between all blocks, respectively, based on the results of Lemma~\ref{lemma2}. 

\begin{lemma}\label{lemma3}
	For block $i \in \llbracket 2,k \rrbracket ,$ we have
	\begin{equation*}
	\mathbb{D}(\tilde{P}_{X_{i-1:i}^NY_{i-1:i}^N\bar{J}_{1}}||\tilde{P}_{X_{i-1}^NY_{i-1}^N\bar{J}_{1}}\tilde{P}_{X_{i}^NY_{i}^N}) \leq \delta_{N}^{(3)}
	\end{equation*}  
	where $\delta_{N}^{(3)}\triangleq {\cal O}(\sqrt[4]{N^{15}\delta_N}).$ 
\end{lemma}
The proof of Lemma~\ref{lemma3} can be found in Appendix~\ref{Appnx:ProofLemma3}.

\begin{lemma}\label{lemma4}
	We have
	\begin{equation*} 
	\mathbb{D}\Big(\tilde{P}_{X_{1:k}^NY_{1:k}^N}||\prod_{i=1}^{k}\tilde{P}_{X_{i}^NY_{i}^N}\Big) \leq (k-1)\delta_{N}^{(3)}
	\end{equation*}
	where $\delta_{N}^{(3)}$ is defined in Lemma 3. 
\end{lemma}
The proof of Lemma~\ref{lemma4} can be found in Appendix~\ref{Appnx:ProofLemma4}.

Finally, by the results of Lemma~\ref{lemma4} we can show in Lemma~\ref{lemma5} that the target distribution $Q_{X^NY^N}$ is approximated asymptotically over all blocks jointly. 
\begin{lemma}\label{lemma5}
	We have
	$$ \mathbb{D}\Big(\tilde{P}_{X_{1:k}^NY_{1:k}^N}||Q_{X^{1:kN}Y^{1:kN}}\Big) \leq \delta_{N}^{(4)}.$$
	where $\delta_{N}^{(4)}\triangleq {\cal O}(k^{3/2}N^{23/8}\delta_{N}^{1/8})$ 
\end{lemma}
\begin{proof} We reuse the proof of \cite[Lemma 5]{chou2016empirical} with substitutions 
	$q_{Y^{1:N}} \leftarrow Q_{X^NY^N},$  $\tilde{p}_{Y_i^{1:N}} \leftarrow \tilde{P}_{Y_i^NX_i^N}.$ 
\end{proof}

\begin{theorem}\label{Thm:PCRR} The polar coding scheme described in Algorithms 1, 2 achieves the region stated in Theorem~\ref{Thm:JointCRR}. It satisfies \eqref{eq:PolarStrngCoorCondtion} for a binary input DMC channel and a target distribution $Q_{XY}$ defined over ${\cal X}\times {\cal Y}$, with an auxiliary random variable $C$ defined over the binary alphabet. 
\end{theorem}

\begin{proof}
The common randomness rate $R_o$ is given as
	\begin{align} \label{eq:Ro}
	\frac{|\bar{J}_1|+|J_{1:k}|}{kN}
	&= \frac{|{\cal V}_{C|XY}|+k|{\cal V}_{C|X}\setminus {\cal V}_{C|XY}|}{Nk} \notag\\
	&= \frac{|{\cal V}_{C|XY}|}{kN} +\frac{|{\cal V}_{C|X}\setminus {\cal V}_{C|XY}|}{N} \notag\\
	&\xrightarrow{N\rightarrow\infty} \frac{H(C|XY)}{k}+ I(Y;C|X)  \notag\\
	&\xrightarrow{k\rightarrow\infty} I(Y;C|X).
	\end{align}
	
The communication rate $R_c$ is given as
	\begin{align}\label{eq:Rc}
	\frac{k|{\cal F}_5 \cup {\cal F}_3|}{kN}& 
	= \frac{k|{\cal V}_{C}\setminus {\cal V}_{C|X}|}{Nk} 
	= \frac{|{\cal V}_{C}\setminus {\cal V}_{C|X}|}{N} \notag\\ 
	&\xrightarrow{N\rightarrow\infty} I(X;C),
	\end{align}
whereas  $R_a$ can be written as
	\begin{align} \label{eq:Ra}
	\frac{|{\cal V}_{A|CXY}|+k|{\cal F}_8|}{kN}& \notag 
	= \frac{|{\cal V}_{A|CXY}|+k|{\cal V}_{A|C}\setminus {\cal V}_{A|CX}|}{kN} \notag\\
	&= \frac{|{\cal V}_{A|CXY}|}{kN} +\frac{|{\cal V}_{A|C}\setminus {\cal V}_{A|CX}|}{N} \notag\\
	&\xrightarrow{N\rightarrow\infty} I(A;X|C) + \frac{H(A|CXY)}{k}\notag\\
	&\xrightarrow{k\rightarrow\infty} I(A;X|C). 
	\end{align}
	
	The rates of local randomness $\rho_1$ and $\rho_2,$ respectively, are given as
	\begin{align}\label{eq:rho1}
	&\rho_1=\frac{k|{\cal F}_6|}{kN} =\frac{|{\cal V}_{A|CX}\setminus {\cal V}_{A|CXY}|}{N} 
	\xrightarrow{N\rightarrow\infty} I(A;Y|CX),  \\
	&\rho_2=\frac{k|{V}_{Y|BC}|}{kN} \xrightarrow{N\rightarrow\infty} H(Y|BC).\label{eq:rho2}
	\end{align}
	
	Finally we see that conditions \eqref{equ:JSRR1}-\eqref{equ:JSRR7} are satisfied by \eqref{eq:Ro}-\eqref{eq:rho2}. Hence, given $R_a$, $R_o$, $R_c$ satisfying Theorem~\ref{Thm:JointCRR}, based on Lemma~\ref{lemma5} and Pinsker's inequality \cite{pinsker1964information} we have \\
	
	\begin{align} \label{eq:resolvProof}
	\mathbb E\big[||\tilde{P}_{X_{1:k}^NY_{1:k}^N}-Q_{X^{1:kN}Y^{1:kN}}||_{{\scriptscriptstyle TV}}\big]  &\leq  \mathbb E\Big[\sqrt{2\mathbb{D}(\tilde{P}_{X_{1:k}^NY_{1:k}^N}||Q_{X^{1:kN}Y^{1:kN}})} \;\Big] \notag\\
	& \leq \sqrt{2\mathbb E \big[\mathbb{D}(\tilde{P}_{X_{1:k}^NY_{1:k}^N}||Q_{X^{1:kN}Y^{1:kN}})\big]}\mathop{\longrightarrow}^{N\rightarrow \infty} 0.
	\end{align} 	
	As a result, from \eqref{eq:resolvProof} there exists an $N\in\mathbb N$ for which the polar code-induced pmf between the pair of actions satisfies the strong coordination condition in \eqref{eq:PolarStrngCoorCondtion}.
\end{proof}

\section{An Example} \label{sec:example}
In the following we compare the performance of the joint scheme in
Section~\ref{sec:JointScheme} and the separation-based scheme in
Section~\ref{sec:SepCoorRE}  using a simple
example. Specifically, we let $X$ to be a Bernoulli-$\frac{1}{2}$ source,
the 
communication channel $P_{B|A}$ to be a BSC with crossover probability $p_o$
(BSC($p_o$)), and the conditional distribution $P_{Y|X}$ to be a BSC($p$). 

\subsection{Basic separation scheme with randomness extraction}	
To derive the rate constraints for the basic separation scheme, we consider  $X-U-Y$ with $U\sim\,$Bernoulli-$\frac{1}{2}$ (which is known to be optimal),  $ P_{U|X}=$BSC$(p_1)$, and  $
P_{Y|U}=$BSC$(p_2)$,  $p_2 \in [0,p]$, $p_1 = \dfrac{p-p_2}{1-2p_2}$.
Using this to obtain the mutual information terms in
Theorem~\ref{Thm:SepSchCRR}, we get
\begin{subequations}
	\begin{align}
	&I(X;U)=1-h_2(p_1),\; I(A;B)=1-h_2(p_o),\label{eq:sep_conda}\\
	&I(XY;U)=1+ h_2(p)-h_2(p_1)-h_2(p_2),\\
	&\text{and } H(Y|U)= h_2(p_2).\label{eq:sep_condc}
	\end{align}
\end{subequations}

After a round of Fourier-Motzkin elimination by using \eqref{eq:sep_conda}-\eqref{eq:sep_condc} in Theorem~\ref{Thm:SepSchCRR}, we obtain the following constraints for the achievable region using the separation-based scheme with randomness extraction:
\begin{subequations}\label{eq:SepSchemObjective}
	\begin{align}
	R_o+\rho_1+\rho_2&\geq h_2(p)-\min\big(h_2(p_2), h_2(p_o)\big),\label{eq:SepSchemObjective1}\\		
	h_2(p_1)&\ge h_2(p_o) \label{eq:SepSchemObjective2}\\
	R_c&\geq 1-h_2(p_1).\label{eq:SepSchemObjective3}
	\end{align}
\end{subequations} 
Note that \eqref{eq:SepSchemObjective1} presents the achievable sum rate
constraint for the required randomness in the system.

\subsection{Joint scheme}
The rate constraints for the joint scheme are constructed in two
stages. First, we derive the scheme for the codebook
cardinalities ${|{\cal A}|=2}$ and ${|{\cal C}|=2}$, an extension to larger
$|{\cal C}|$ is straightforward but more tedious (see Figs.~\ref{fig:MinSumRateComparsion} and
\ref{fig:CommunicationRateComparision})\footnote{Note that these cardinalities are not optimal. They are, however, analytically feasible and provide a good intuition about the performance of the scheme.}. The joint scheme correlates the codebooks while ensuring that the
decodability constraint \eqref{equ:JSRR5} is satisfied. To find the best tradeoff between these two features, we find the joint distribution $P_{AC}$ that maximizes $I(B;C)$. For ${|\mathcal{C}|=2}$ this is simply given by
${P_{A|C}=\delta_{ac}}$, where $\delta_{ac}$ denotes the Kronecker delta. 
Then, the distribution $P_{X}P_{CA|X}P_{B|A}P_{Y|BC}$
that produces the boundary of the strong coordination region for the joint
scheme is formed by cascading two BSCs and another symmetric  channel,
yielding the Markov chain ${X-(C,A)-(C,B)-Y}$, with the channel transition matrices
\begin{align}
P_{CA|X}&=\left[ \begin{matrix}
1-p_1 &0  & 0 & p_1 \\
p_1  &0    & 0 & 1-p_1 
\end{matrix}\right],\\
P_{CB|CA}&=\left[\begin{matrix}
1-p_o  &p_o  & 0   & 0\\
0   &0    & p_o & 1-p_o 
\end{matrix}\right], \\
P_{Y|CB}&=\left[\begin{matrix}
1-\alpha  &  1-\beta & \beta & \alpha \\
\alpha  & \beta & 1-\beta & 1-\alpha
\end{matrix}\right]^T
\end{align} 
for some $\alpha,\beta \in [0,1].$

Then, the mutual information terms  in
Theorem~\ref{Thm:JointCRR}   can be 
expressed with $p_2\triangleq (1-p_o)\alpha + p_o\beta$ as
\begin{align*}
I(X;AC) &=I(X;C)= 1-h_2(p_1),&\\
I(XY;AC) &=I(XY;C) =1+ h_2(p)-h_2(p_1)-h_2(p_2),&\\
I(B;C) &= 1-h_2(p_o),\text{ and }\\
H(Y|BC)&= p_oh_2(\beta)+(1-\!p_o)h_2(\alpha).
\end{align*}
To find the minimum achievable sum rate we first perform Fourier-Motzkin
elimination on the rate constraints in Theorem~\ref{Thm:JointCRR} and then
minimize the information terms with respect to the parameters $p_2$,
$\alpha$, and $\beta$ as follows:
\begin{align} \label{eq:JointSchemObjective}
&R_o+\rho_1+\rho_2=\min_{p_2,\alpha,\beta} \big(h_2(p)-h_2(p_2)+(1-p_o)h_2(\alpha)+p_oh_2(\beta)\big) \\
&\qquad\qquad\quad\,\,\,\, \text{subject to  }\,\begin{array}{rcl} h_2(p_1)&>&h_2(p_o), \notag\\
R_c&\geq& 1-h_2(p_1), \notag\\
p&=&p_1-2p_1p_2+p_2.\end{array}
\end{align}

\subsection{Numerical  results}
Fig.~\ref{fig:MinSumRateComparsion} presents a comparison between the
minimum randomness sum rate $R_o+\rho_1+\rho_2$ required to achieve coordination using both the joint and the separate scheme with randomness extraction. The communication channel is given by BSC($p_o$), and 
the target distribution is set as
$Q_{Y|X}=\mathrm{BSC}(0.4)$. The rates for the joint scheme are obtained
by solving the optimization problem in \eqref{eq:JointSchemObjective}.
Similar results are obtained for the joint scheme with  $|{\cal C}|> 2$. For
the separate scheme we choose $p_2$ such that {$h_2(p_1)=h_2(p_0)$} to maximize
the amount of extracted randomness. We
also include  the performance of the separate scheme
without randomness extraction. 

\begin{figure}[th!]
	\centering
	\begin{subfigure}{.5\textwidth}
		\centering
		\includegraphics[scale=0.41]{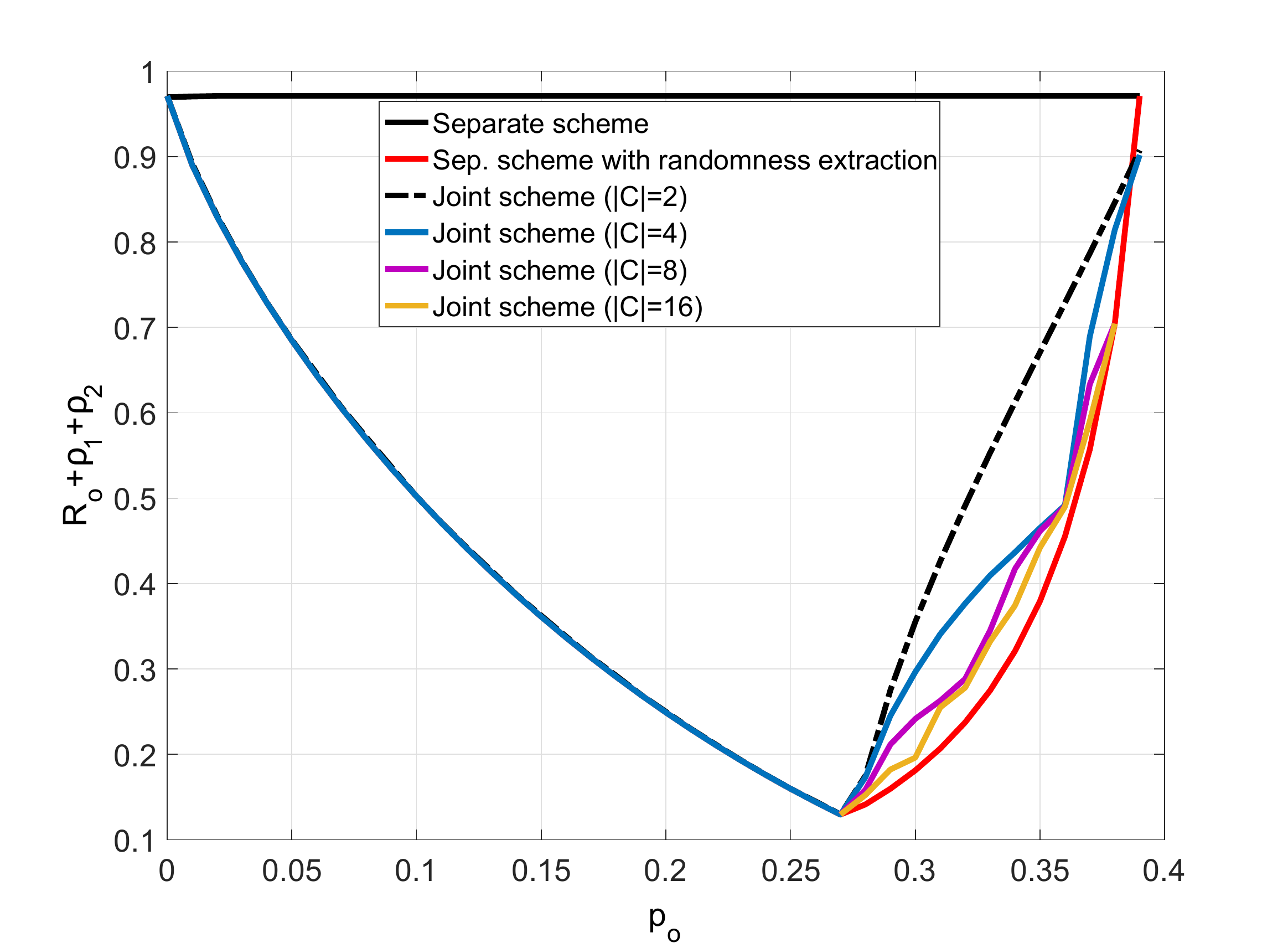}
		\caption{Randomness sum rate vs.~BSC crossover probability $p_0$.}
		\label{fig:MinSumRateComparsion}
	\end{subfigure}%
	\begin{subfigure}{.5\textwidth}
		\centering
		\includegraphics[scale=0.422]{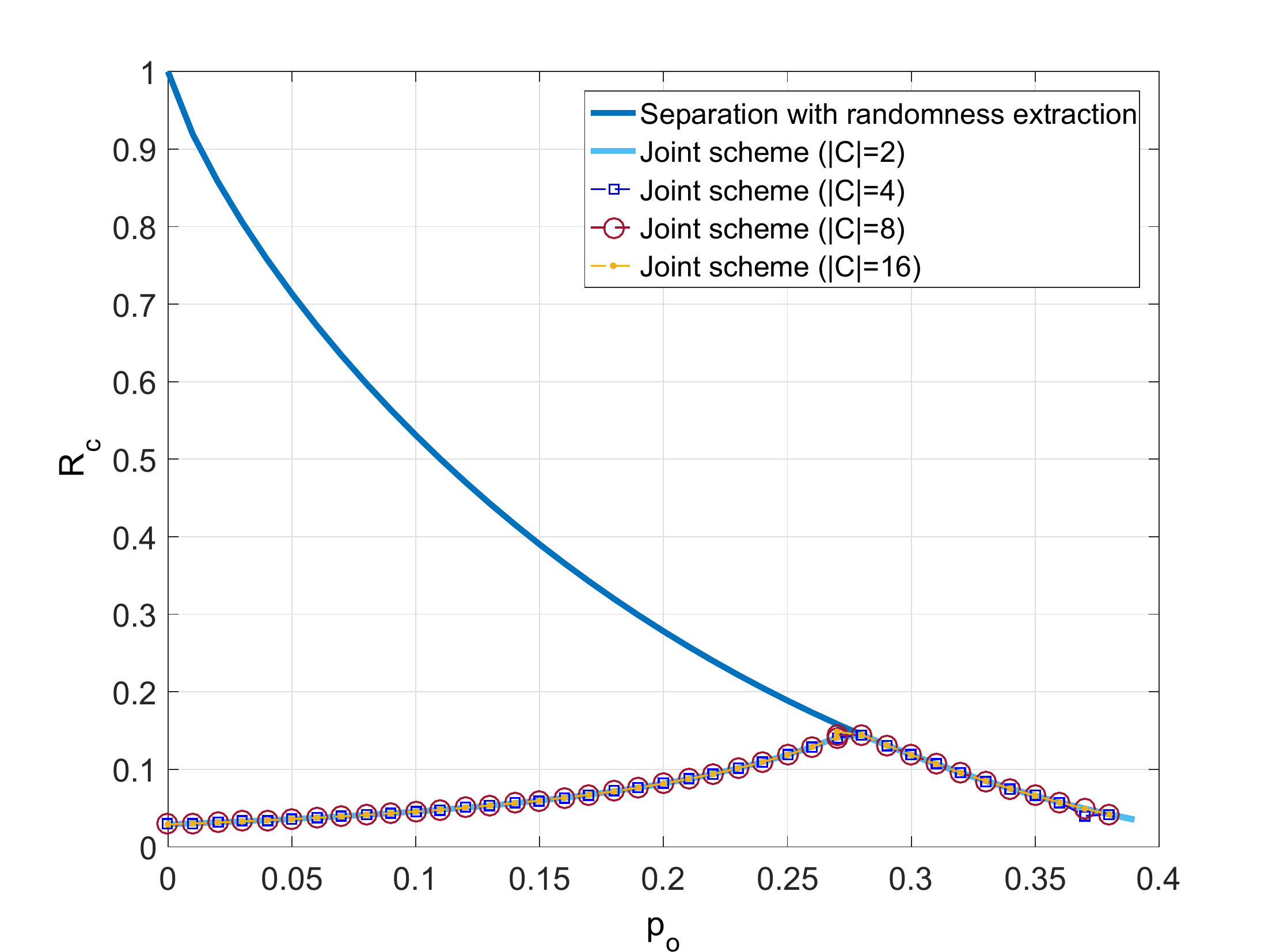}
		\caption{Communication rate vs.~BSC crossover probability $p_0$.}
		\label{fig:CommunicationRateComparision}
	\end{subfigure}%
	\caption{Numerical results for target distribution $Q_{Y|X}=\mathrm{BSC}(0.4)$}
\end{figure}

\vspace{1ex}
As can be seen from Fig.~\ref{fig:MinSumRateComparsion}, both the joint scheme and the separate
scheme with randomness extraction provide the same sum rate
$R_o+\rho_1+\rho_2$ for $p_o\leq p'_o$ where
$p'_o\triangleq\frac{1-\sqrt{1-2p}}{2}$. We also observe that for noisier
channels with $p_o >p'_o$ the joint scheme approaches the performance of the separate scheme
when the cardinality of $C$ is increased. The increase of
$R_o+\rho_1+\rho_2$ for $p_o> p'_o$ is due to the fact that
in this regime the channel provides more than sufficient randomness for simulating the
action sequence $Y^n$ via the test channel $P_{Y|BC}$ (see
Fig.~\ref{fig:StrongCoordination}). As a result, the parameters
$\alpha$ and $\beta$ associated with  $P_{Y|BC}$ must be adjusted to
ensure that \eqref{eq:StrngCoorCondtion} is still satisfied. As $p_o$ increases further, the required total randomness of the joint scheme approaches the one for the basic separate scheme again.

Fig.~\ref{fig:CommunicationRateComparision} provides a comparison of the
communication rate for both schemes. Note that the joint scheme provides
significantly smaller rates than the separation scheme with randomness
extraction  for $p_o\leq p'_o$, independently
of the cardinality of $|\mathcal{C}|$.  Thus, in this regime joint
coordination-channel coding provides an advantage in terms of communication
cost and outperforms a separation-based scheme for the same amount of
randomness injected into the system. 

\vspace*{0.2ex}
\section{summary and concluding remarks}
In this paper, we have investigated a fundamental question regarding
  communication-based coordination: Is separate coordination and channel
  coding optimal in the context of point-to-point strong coordination? In particular, we considered a two-node strong coordination setup with a DMC as the communication link. 
To that extent, we presented achievability results for this setting and constructed a general joint
coordination-channel coding scheme based on random codes. We also
  provided a capacity result  for two special cases of the noisy strong
  coordination setup: the first case considers the actions at Node $\mathsf
  Y$ to be determined by Node $\mathsf X$ and the second case addresses the
  case where the discrete memoryless communication channel is a
  deterministic channel. The proof technique underlying our joint coding
scheme is based on channel resolvability, a technique which is widely used
in analyzing strong coordination problems. In addition, we presented a
separation-based benchmark scheme and utilized randomness extraction to improve its performance. In this scheme, randomness is extracted from the channel at the decoder when the channel is an additive noise DMC. In addition, we have also proposed a constructive coding scheme based on polar codes for the noisy two-node network that can achieve all the rates that the joint scheme can, where achievability is guaranteed asymptotically.

Finally, by leveraging our random coding results, we presented an
example for coordinating a doubly binary symmetric source over a binary
symmetric communication channel in which the proposed joint scheme
outperforms a separation-based scheme in terms of achievable
  communication rate. As a result, we conclude that a separation-based scheme, even if it exploits randomness extraction from the communication channel, is sub-optimal for this problem. Although this work yields some insight in coordination over noisy communication links, a general converse proof to establish the optimality of the presented joint coding scheme is still open.  	

\vspace*{0.5ex}
\appendices
\section{Proof of Lemma~\ref{Lma:Secrecy}}\label{Appnx:ProofSecrecyLemma}
The proof of Lemma~3 leverages the results from Section~\ref{sec:DecoCons}. The bound on $||\check{P}_{X^nJ}-Q^n_{X}P_J||_{\scriptscriptstyle TV}$ is obtained in a similar manner as in the proof of Lemma~\ref{Lma:Reslv}. Note that here we also drop the subscripts from the pmfs for simplicity, e.g., $P_{X|AC}^n(x^n|A_{ijk}^n, C_{ij}^n)$ will be denoted by $P(x^n|A_{ijk}^n, C_{ij}^n),$ and $Q^n_{X}(x^n)$ will be denoted by $Q(x^n)$ in the following.
\begin{proof}[Proof of Lemma~\ref{Lma:Secrecy}]
	\begin{equation*}
	\begin{split}
	\mathbb E_\mathsf{C} [\mathbb{D}(\check{P}_{X^nJ}||&Q^n_{X}P_J)]\\
	&= \mathbb E_\mathsf{C}\Bigg[\sum_{x^n,j} \Big( \sum_{i,k} \dfrac{P(x^n|A^n_{ijk},C^n_{ij})}{2^{nR}}\Big) \log \Big( \sum_{i',k'} \dfrac{P(x^n|A^n_{i'j'k'},C^n_{i'j'})}{2^{nR}Q(x^n)P(j)}\Big)\Bigg]\\
	&=\sum_{x^n} \sum_{i,j,k}  \mathbb E_\mathsf{C} \Bigg[ \Big( \dfrac{ P(x^n|A^n_{ijk},C^n_{ij})}{2^{nR}}\Big) \log \Big( \sum_{i',k'} \dfrac{P(x^n|A^n_{i'j'k'},C^n_{i'j'})}{2^{n(R_a+R_c)}Q(x^n)}\Big)\Bigg] \\
	&\stackrel{(a)}{=} \sum_{x^n} \sum_{i,j,k}  \mathbb E_{A^n_{ijk},C^n_{ij}}\Bigg[ \Big(  \dfrac{P(x^n|A^n_{ijk},C^n_{ij})}{2^{nR}}\Big) \mathbb E_{rest} \Big[ \log \Big( \sum_{i',k'} \dfrac{P(x^n|A^n_{i'j'k'},C^n_{i'j'})}{2^{n(R_a+R_c)}Q(x^n)}\Big) \Big| A^n_{ijk}C^n_{ij}\Big]\Bigg] \\
	&\stackrel{(b)}{\leq} \sum_{x^n} \sum_{i,j,k} \mathbb E_{A^n_{ijk}C^n_{ij}}\Bigg[ \Big( \dfrac{ P(x^n|A^n_{ijk},C^n_{ij})}{2^{nR}}\Big)\log \Big( \mathbb E_{rest} \Big[\sum_{i',k'} \dfrac{ P(x^n|A^n_{i'j'k'},C^n_{i'j'})}{2^{n(R_a+R_c)}Q(x^n)} \Big| A^n_{ijk}C^n_{ij} \Big] \Big)\Bigg] 
		\end{split}
	\end{equation*}
	\begin{equation*}
	\begin{split}
	&\stackrel{(c)}{=} \sum_{x^n} \sum_{a^n_{ijk},c^n_{ij}} \sum_{i,j,k} \dfrac{P(x^n,a^n_{ijk},c^n_{ij})}{2^{nR}} \log \Bigg( \sum_{\substack{i',k':\\(i',j',k')=(i,j,k)}} \mathbb E_{A^n_{ijk}C^n_{ij}} \Big[ \dfrac{ P(x^n|A^n_{i'j'k'},C^n_{i'j'})}{2^{n(R_a+R_c)}Q(x^n)}\Big|A^n_{ijk}C^n_{ij} \Big]\\
	&\hspace*{37ex}+\sum_{\substack{i',k':\\(i',j')=(i,j),(k'\neq k)}} \mathbb E_{A^n_{ijk}C^n_{ij}} \Big[ \dfrac{P(x^n|A^n_{i'j'k'},C^n_{i'j'})}{2^{n(R_a+R_c)}Q(x^n)}\Big|A^n_{ijk}C^n_{ij} \Big]\\ 
	&\hspace*{37ex}+\sum_{\substack{i',j',k':\\(i',j',k')\neq(i,j,k)}} \mathbb E_{A^n_{ijk}C^n_{ij}} \Big[ \dfrac{ P(x^n|A^n_{i'j'k'},C^n_{i'j'})}{2^{n(R_a+R_c)}Q(x^n)}\Big|A^n_{ijk}C^n_{ij}\Big] \Bigg) \\
	&\stackrel{(d)}{=} \sum_{x^n} \sum_{a^n_{ijk},c^n_{ij}} \sum_{i,j,k} \dfrac{P(x^n,a^n_{ijk},c^n_{ij})}{2^{nR}}  \log \Bigg(\dfrac{P(x^n|a^n_{ijk},c^n_{ij})}{2^{n(R_a+R_c)}Q(x^n)} + \sum_{\substack{i',k':\\(i',j')=(i,j),(k'\neq k)}} \dfrac{P(x^n|c^n_{ij})}{2^{n(R_a+R_c)}Q(x^n)} \\
	&\hspace*{37ex}+ \sum_{\substack{i',k':\\(i',j',k')\neq(i,j,k)}} \dfrac{P(x^n)}{2^{n(R_a+R_c)}Q(x^n)} \Bigg) \\
	&\stackrel{(e)}{\leq}\sum_{x^n} \sum_{a^n_{ijk},c^n_{ij}}  P(x^n a^n_{ijk},c^n_{ij}) \log \Bigg (\dfrac{P(x^n|a^n_{ijk},c^n_{ij})}{2^{n(R_a+R_c)}Q(x^n)} +(2^{nR_a}) \dfrac{P(x^n|c^n_{ij})}{2^{n(R_a+R_c)}Q(x^n)} +1 \Bigg)\\
	&\stackrel{(f)}{\leq} \Bigg[ \sum_{ (x^n,a^n,c^n)\in {\cal T}_\epsilon^n(P_{XAC})}\! P(x^n,a^n,c^n) \log \Bigg (\dfrac{2^{-nH(X|AC)(1-\epsilon)}} {2^{n(R_a+R_o)}2^{-nH(X)(1+\epsilon)}} + \dfrac{2^{-nH(X|C)(1-\epsilon)}}{2^{nR_c}2^{-nH(X)(1+\epsilon)}} + 1 \Bigg)\Bigg]\\ 
	&\hspace*{52ex}+ \mathbb P\big((x^n,a^n,c^n) \notin {\cal T}_\epsilon^n( P_{XAC}) \big)\log(2\mu_{X}^{-n}+1)\\
	&\stackrel{(g)}{\leq}\Bigg[ \sum_{(x^n,a^n,c^n)\in {\cal T}_\epsilon^n( P_{XAC})} P(x^n,a^n,c^n) \log \Bigg(\dfrac{2^{n(I(X;AC)+\delta(\epsilon))}}{2^{n(R_c+R_a)}}+ \dfrac{2^{n(I(X;C)+\delta(\epsilon))}}{2^{n(R_c)}}+1 \Bigg) \Bigg]\\ 
	&\hspace*{52ex} + \big(2{\cal |X||A||C|}e^{-n\epsilon^2\mu_{XAC}}\big)\log(2\mu_{X}^{-n}+1)\\ 
	&\stackrel{(h)}{\leq} \epsilon',
	\end{split}
	\end{equation*}
	\noindent where $\epsilon'>0,$ $\epsilon' \rightarrow 0$ as $n\rightarrow \infty$.\\
	($a$) follows from the law of iterated expectations. Note that we have used $(a^{n}_{ijk},c^{n}_{ij})$ to denote the codewords corresponding to the indices $(i,j,k)$ and $(a^{n}_{i'j'k'},c^{n}_{i'j'})$ to denote the codewords corresponding to the indices $(i',j',k')$, respectively.\\
	($b$) follows from Jensen's inequality \cite{EoIT:2006};\\
	($c$) follows from dividing the inner summation over the indices $(i',k')$ into three subsets based on the indices $(i,j,k)$ from the outer summation;\\
	($d$) results from taking the conditional expectation within the subsets in (c);\\  
	($e$) follows from
	\begin{itemize}
		\item $(i',j',k')=(i,j,k)$: there is only one pair of codewords represented by the indices $A^{n}_{ijk},C^{n}_{ij}$ corresponding to $x^n$;
		\item when $(k'\neq k)$ and $(i',j')=(i,j)$ there are $(2^{nR_a}-1)$ indices in the sum; 
		\item $(i',j',k')\neq(i,j,k)$: the number of the indices is at most $2^{n(R_a+R_c)}.$ Moreover, $P(x^n)$ is less than $\epsilon$ close to $Q(x^n)$ as a consequence of Lemma~\ref{Lma:Reslv} and \cite[Lemma 16]{cuff2009CiN4CB}.	
	\end{itemize} 
	($f$) results from splitting the outer summation: The first summation contains typical sequences and is bounded by using the probabilities of the typical set. The second summation contains the tuple of sequences when the action sequence $x^n$ and codewords $c^n,a^n$, represented here by the indices $(i,j,k)$, are not $\epsilon$-jointly typical (i.e.,~$(x^n,a^n,c^n)\notin {\cal T}_\epsilon^n( P_{XAC})$). This sum is upper bounded following \cite{BK14} with $\mu_{X} =\min^*_{x} (P_{X})$.\\ 
	($g$) follows from the Chernoff bound on the probability that a sequence is not strongly typical \cite{kramer2008TinMUIT}.\\
	($h$) consequently, the contribution of typical sequences can be asymptotically made small if 
	$$	R_a+R_c \geq I(X;AC),\quad R_c\geq I(X;C).$$
	
	The second term converges to zero exponentially fast with $n$ \cite{kramer2008TinMUIT}, and following Pinsker's inequality \cite{pinsker1964information} we have  
	\begin{align*} 
	\mathbb E_\mathsf{C} \big[||\check{P}_{X^nJ}-Q^n_{X}P_J||_{\scriptscriptstyle TV}\big]  
	&\leq \mathbb E_\mathsf{C} \Big[\sqrt{2\mathbb{D}(\check{P}_{X^nJ}||Q^n_{X}P_J)}\;\Big] \notag\\
	&\leq \sqrt{2\mathbb E_\mathsf{C} \big[\mathbb{D}(\check{P}_{X^nJ}||Q^n_{X}P_J)\big]} \leq \sqrt{2\epsilon'}.
	\end{align*} 
\end{proof}

\section{Proof of Lemma~\ref{lemma3}}\label{Appnx:ProofLemma3}
We reuse the proof of \cite[Lemma 3]{chou2016empirical} with substitutions $q_{U^{1:N}} \leftarrow Q_{C^N},$ $q_{Y^{1:N}} \leftarrow Q_{X^N\!Y^N},$ $\tilde{p}_{U_i^{1:N}} \leftarrow \tilde{P}_{C_i^N},$ ${\tilde{p}_{Y_i^{1:N}}\!\leftarrow\!\tilde{P}_{Y_i^N\!X_i^N}},$ and $\bar{R}_1 \leftarrow \bar{J}_{1}$. This results in the Markov chain $X_{i-1}^N\!\widetilde{Y}_{i-1}^N-\bar{J}_{1}-X_i^N\!\widetilde{Y}_i^N$ replacing the chain in~\cite[Lemma~3]{chou2016empirical}.
\begin{proof}[Proof of Lemma~\ref{lemma3}]
\begin{equation}
\begin{split}\label{lemma3Part1}
	&H(U_{2}^N[{\cal V}_{C|XY}]|X^NY^N)- H(\widetilde{U}^N_{2_i}[{\cal V}_{C|XY}]|{X}_i^N\widetilde{Y}_i^N)\\
	&= H(U_{2}^N[{\cal V}_{C|XY}]X^NY^N)- H(\widetilde{U}^N_{2_i}[{\cal V}_{C|XY}]{X}_i^N\widetilde{Y}_i^N)- H(X^NY^N)+H({X}_i^N\widetilde{Y}_i^N)\\
	&\stackrel{(a)}{\leq}  \mathbb{D}(\tilde{P}_{U_{2_i}^N[{\cal V}_{C|XY}]X_i^NY_i^N} ||Q_{U_{2}^N[{\cal V}_{C|XY}]X^NY^N}) + N^3\log(|{\cal X}||{\cal Y}||{\cal C}|)\sqrt{2\ln 2}\sqrt{\mathbb{D}(\tilde{P}_{U_{2_i}^N[{\cal V}_{C|XY}]X_i^NY_i^N} ||Q_{U_{2}^N[{\cal V}_{C|XY}]X^NY^N})}\\
	&\hspace*{18ex} + \mathbb{D}(Q_{X^NY^N}||\tilde{P}_{X_i^NY_i^N}) + N^2\log(|{\cal X}||{\cal Y}|)\sqrt{2\ln 2}\sqrt{\mathbb{D}(\tilde{P}_{X_i^NY_i^N})||Q_{X^NY^N}}\\
	&\stackrel{(b)}{\leq}  \mathbb{D}(\tilde{P}_{U_{2_i}^NX_i^NY_i^N} ||Q_{U_{2}^NX^NY^N}) + N^3\log(|{\cal X}||{\cal Y}||{\cal C}|)\sqrt{2\ln 2}\sqrt{\mathbb{D}(\tilde{P}_{U_{2_i}^NX_i^NY_i^N} ||Q_{U_{2}^NX^NY^N})}\\
	&\hspace*{18ex} + \mathbb{D}(Q_{X^NY^N}||\tilde{P}_{X_i^NY_i^N}) + N^2\log(|{\cal X}||{\cal Y}|)\sqrt{2\ln 2}\sqrt{\mathbb{D}(\tilde{P}_{X_i^NY_i^N})||Q_{X^NY^N}}\\
	&\stackrel{(c)}{\leq}  \delta_{N}^{(1)}+ N^3\log(|{\cal X}||{\cal Y}||{\cal C}|)\sqrt{2\ln 2}\sqrt{\delta_{N}^{(1)}} -N\log(\mu_{XY})\sqrt{2\ln 2} \sqrt{\mathbb{D}(\tilde{P}_{X_i^NY_i^N}||Q_{X^NY^N})} \\
	&\hspace*{18ex}  + N^2\log(|{\cal X}||{\cal Y}|)\sqrt{2\ln 2}\sqrt{\mathbb{D}(\tilde{P}_{X_i^NY_i^N}||Q_{X^NY^N})}\\
	&\leq  \delta_{N}^{(1)}+ N^3\log(|{\cal X}||{\cal Y}||{\cal C}|)\sqrt{2\ln 2}\sqrt{\delta_{N}^{(1)}} -N\log(\mu_{XY})\sqrt{2\ln 2} \sqrt{\delta_{N}^{(1)}} + N^2\log(|{\cal X}||{\cal Y}|)\sqrt{2\ln 2}\sqrt{\delta_{N}^{(1)}}\\
	&\leq \hat{\delta}_{N}^{(3)}, 
\end{split}
\end{equation}

\noindent where \\
($a$) follows from \cite[Lemma 17]{chou2016empirical};\\
($b$) follows from the chain rule of KL divergence \cite{EoIT:2006};\\
($c$) follows from Lemma~\ref{lemma2} and \cite[Lemma 14]{chou2016empirical}.\\
\noindent Hence, for block $i \in \llbracket 2,k \rrbracket ,$ we have
\begin{equation*}
\begin{split}
\mathbb{D}(\tilde{P}_{X_{i-1:i}^NY_{i-1:i}^N\bar{J}_1} || \tilde{P}_{X_{i-1}^NY_{i-1}^N\bar{J}_1}\tilde{P}_{X_{i}^NY_{i}^N}) 
&=I({X}_{i-1}^N\widetilde{Y}_{i-1}^N\bar{J}_1; {X}_{i}^N\widetilde{Y}_{i}^N)\\
&=I({X}_{i}^N\widetilde{Y}_{i}^N;\bar{J}_1)+I({X}_{i-1}^N\widetilde{Y}_{i-1}^N; {X}_{i}^N\widetilde{Y}_{i}^N|\bar{J}_1)\\
&=I({X}_{i}^N\widetilde{Y}_{i}^N;\bar{J}_1)\\
&=I({X}_{i}^N\widetilde{Y}_{i}^N;U_{2_i}^N[{\cal V}_{C|XY}])\\
&= H(U_{2_i}^N[{\cal V}_{C|XY}])- H(\widetilde{U}^N_{2_i}[{\cal V}_{C|XY}]|{X}_{i}^N\widetilde{Y}_{i}^N)\\
&\stackrel{(a)}{\leq}  |{\cal V}_{C|XY}|\log (|{\cal C}|) -H(U_{2}^N[{\cal V}_{C|XY}]|X^NY^N)+ \hat{\delta}_{N}^{(3)}\\
&\stackrel{(b)}{\leq}  |{\cal V}_{C|XY}| - \sum_{j \in {\cal V}_{C|XY}} H(U_{2,j}|X^NY^NU^{j-1})+ \hat{\delta}_{N}^{(3)}\\
\end{split}
\end{equation*}
\begin{equation*}
\begin{split} 
&\leq |{\cal V}_{C|XY}| - |{\cal V}_{C|XY}|(1-\delta_{N}) + \hat{\delta}_{N}^{(3)}\\
&= |{\cal V}_{C|XY}| \delta_{N} + \hat{\delta}_{N}^{(3)}\\
&\leq N\delta_{N} + \hat{\delta}_{N}^{(3)}\\
&\leq {\delta}_{N}^{(3)},
\end{split}
\end{equation*}
\noindent where ($a$) follows from \eqref{lemma3Part1}, and ($b$) follows from the definition of the high entropy sets \eqref{eq:C_Sets}.
\end{proof}

\section{Proof of Lemma~\ref{lemma4}}\label{Appnx:ProofLemma4}
We reuse the proof of \cite[Lemma~4]{chou2016empirical} with substitutions $\tilde{p}_{Y_i^{1:N}} \leftarrow \tilde{P}_{X_i^NY_i^N},$ and $\bar{R}_1 \leftarrow \bar{J}_{1}$. This will result in the Markov chain $X_{1:i-2}^N\widetilde{Y}_{1:i-2}^N-\bar{J}_{1}X_{i-1}^N\widetilde{Y}_{i-1}^N-X_i^N\widetilde{Y}_i^N$, replacing the chain in \cite[Lemma~4]{chou2016empirical}. 
\begin{proof}[Proof of Lemma~\ref{lemma4}]
	\begin{equation*}
	\begin{split}
	\mathbb{D}\Big(\tilde{P}_{X_{1:k}^NY_{1:k}^N}||\prod_{i=1}^{k}\tilde{P}_{X_{i}^NY_{i}^N}\Big)
	&\stackrel{(a)}{=} \sum_{i=2}^{k} I({X}_{i}^N\widetilde{Y}_{i}^N;{X}_{1:i-1}^N\widetilde{Y}_{1:i-1}^N)\\
	&{\leq} \sum_{i=2}^{k} I({X}_{i}^N\widetilde{Y}_{i}^N;{X}_{1:i-1}^N\widetilde{Y}_{1:i-1}^N\bar{J}_1) \\
	&{=}\sum_{i=2}^{k} I({X}_{i}^N\widetilde{Y}_{i}^N;{X}_{i-1}^N\widetilde{Y}_{i-1}^N\bar{J}_1) +  I({X}_{i}^N\widetilde{Y}_{i}^N;{X}_{2:i-1}^N\widetilde{Y}_{2:i-1}^N|{X}_{i-1}^N\widetilde{Y}_{i-1}^N\bar{J}_1)  \\
	&\stackrel{(b)}{=} \sum_{i=2}^{k} \mathbb{D}(\tilde{P}_{X_{i-1:i}^NY_{i-1:i}^N\bar{J}_1} || \tilde{P}_{X_{i-1}^NY_{i-1}^N\bar{J}_1}\tilde{P}_{X_{i}^NY_{i}^N})\\
    &\stackrel{(c)}{\leq}  \sum_{i=2}^{k}\delta_{N}^{(3)}\\ 
	&= (k-1)\delta_{N}^{(3)},
	\end{split}
	\end{equation*}
	\noindent where \\
	($a$) follows from \cite[Lemma 15]{chou2016empirical}.\\
	($b$) holds by the Markov chain $X_{1:i-2}^N\widetilde{Y}_{1:i-2}^N-\bar{J}_{1}X_{i-1}^N\widetilde{Y}_{i-1}^N-X_i^N\widetilde{Y}_i^N$. \\
	($c$) follows from Lemma~\ref{lemma3}.
\end{proof}

\bibliographystyle{IEEEtran}
\bibliography{IEEEabrv,references}
\end{document}